\documentclass[preprint,authoryear]{elsarticle}
\usepackage{amsmath}
\usepackage{mathtools}

\biboptions{sort,longnamesfirst}

\newtheorem{theorem}{Theorem}
\newtheorem{lemma}[theorem]{Lemma}
\newtheorem{proposition}[theorem]{Proposition}
\newtheorem{corollary}[theorem]{Corollary}
\newdefinition{definition}[theorem]{Definition}
\newdefinition{remark}[theorem]{Remark}
\newdefinition{example}[theorem]{Example}
\newproof{proof}{Proof}

\newcommand{\R}{\mathbb{R}}
\newcommand{\Q}{\mathbb{Q}}
\newcommand{\N}{\mathbb{N}}

\newcommand{\K}{\mathbb{K}}
\newcommand{\MAT}[3]{M_{#1\ifthenelse{\equal{#2}{}}{}{,#2}}\ifthenelse{\equal{#3}{}}{}{\left(#3\right)}}
\newcommand{\Rgen}{\R_{G}}
\newcommand{\Rpoly}{\R_{P}}
\newcommand{\Rp}{\R_{+}}

\newcommand{\Rps}{\R_{+}^{*}}

\newcommand{\dom}{\operatorname{dom}}

\newcommand{\inorm}[2]{\left\lVert{#1}\right\rVert_{#2}}

\newcommand{\infnorm}[1]{\inorm{#1}{}}

\newcommand{\sigmap}[1]{{\Sigma{#1}}}
\newcommand{\degp}[1]{{\operatorname{deg}(#1)}}
\newcommand{\poly}{\operatorname{poly}}

\newcommand{\sgn}[1]{\operatorname{sgn}(#1)}

\newcommand{\indicator}[1]{\mathds{1}_{#1}}
\newcommand{\idfun}{\operatorname{id}}

\newcommand{\myclass}[1]{\operatorname{#1}}
 \newcommand{\gval}[2][]{\ensuremath{\myclass{GVAL}[#2]}}
 \newcommand{\gpval}[1][]{\ensuremath{\myclass{GPVAL}}}
 \newcommand{\gc}[3][]{\ensuremath{\myclass{ATSC}(#2,#3)}}
 \newcommand{\gpc}[1][]{\ensuremath{\myclass{ATSP}}}
\newcommand{\cgc}[1][]{\ensuremath{\myclass{ATSC}}}
\newcommand{\gexpc}[1][]{\ensuremath{\myclass{AEXP}}}
\newcommand{\gwc}[2]{\ensuremath{\myclass{AW}(#1,#2)}}
\newcommand{\gpwc}{\ensuremath{\myclass{AWP}}}
\newcommand{\cgwc}{\ensuremath{\myclass{AWC}}}
\newcommand{\grc}[3]{\ensuremath{\myclass{ARC}(#1,#2,#3)}}
\newcommand{\gprc}{\ensuremath{\myclass{ARP}}}
\newcommand{\gsc}[3]{\ensuremath{\myclass{AS}(#1,#2,#3)}}
\newcommand{\gpsc}{\ensuremath{\myclass{ASP}}}
\newcommand{\goc}[3]{\ensuremath{\myclass{AOC}(#1,#2,#3)}}
\newcommand{\gpoc}{\ensuremath{\myclass{AOP}}}
\newcommand{\cgoc}{\ensuremath{\myclass{AOC}}}
\newcommand{\guc}[4]{\ensuremath{\myclass{AXC}(#1,#2,#3,#4)}}
\newcommand{\gpuc}{\ensuremath{\myclass{AXP}}}
\newcommand{\glc}[2][]{\ensuremath{\myclass{ALC}(#2)}}
\newcommand{\gplc}[1][]{\ensuremath{\myclass{ALP}}}
\newcommand{\cglc}[1][]{\ensuremath{\myclass{ALC}}}

\newcommand{\Unaware}{Extreme}
\newcommand{\unaware}{extreme}
\newcommand{\unawarely}{extremely}

\newcommand{\jacobian}[1]{J_{#1}}
\newcommand{\grad}[1]{\nabla{#1}}
\newcommand{\scalarprod}[2]{{#1}\cdot{#2}}
\newcommand{\bigO}[1]{\mathcal{O}\left(#1\right)}

\newcommand{\transpose}[1]{{#1}^T}
\newcommand{\pastsup}[2]{{\sup}_{#1}#2}

\newcommand{\mtt}[1]{\mathtt{#1}}
\newcommand{\ovl}[1]{\overline{#1}}

\newcommand{\orbit}[1]{\mathcal{O}_x}

\newcommand{\PIVP}{{PIVP}}

\newcommand{\myop}[1]{\operatorname{#1}}

\newcommand{\lxh}{\myop{lxh}}
\newcommand{\hxl}{\myop{hxl}}
\newcommand{\plil}{\myop{plil}}
\newcommand{\reach}{\myop{reach}}

\newcommand{\norm}{\myop{norm}}
\newcommand{\mx}{\myop{mx}}
\newcommand{\sample}{\myop{sample}}

\newcommand{\glen}[1]{\myop{len}_{#1}}

\usepackage[english]{babel}
\usepackage[T1]{fontenc}
\usepackage[utf8]{inputenc}
\usepackage{amssymb,amsfonts}
\usepackage{libertine}
\usepackage{mathdots}
\usepackage{color}
\usepackage{xifthen}
\usepackage{stmaryrd}
\usepackage{tikz}
\usepgflibrary{arrows}
\usetikzlibrary{fit}
\usetikzlibrary{patterns}
\usetikzlibrary{shapes.geometric}
\usepackage{longtable}
\usepackage{dsfont}
\usepackage[
  pdftex,
  bookmarks=true,
  bookmarksnumbered=true,
  hypertexnames=false,
  breaklinks=true,
  linkbordercolor={0 0 1}
  ]{hyperref}

\allowdisplaybreaks

\definecolor{mygreen}{rgb}{0,.5,0}
\colorlet{myyellow}{yellow!80!blue}

\newcommand{\defref}[1]{Definition~\ref{#1}}
\newcommand{\thref}[1]{Theorem~\ref{#1}}
\newcommand{\lemref}[1]{Lemma~\ref{#1}}
\newcommand{\propref}[1]{Proposition~\ref{#1}}
\newcommand{\figref}[1]{Figure~\ref{#1}}
\newcommand{\secref}[1]{Section~\ref{#1}}

\begin{document}

\title{Computing with Polynomial Ordinary Differential Equations}

\author[lix]{Olivier Bournez\corref{cor}\fnref{dga}}
\ead{bournez@lix.polytechnique.fr}
\author[fct,sqig]{Daniel Gra\c{c}a\fnref{feder}}
\ead{dgraca@ualg.pt}
\author[lix]{Amaury Pouly\fnref{dga}}
\ead{pamaury@lix.polytechnique.fr}

\address[lix]{École Polytechnique, LIX, 91128 Palaiseau Cedex, France}
\address[fct]{CEDMES/FCT, Universidade do Algarve, C. Gambelas, 8005-139 Faro, Portugal}
\address[sqig]{SQIG/Instituto de Telecomunica\c{c}\~{o}es, Lisbon, Portugal}
\cortext[cor]{Corresponding author}
\fntext[feder]{Daniel Gra\c{c}a was partially supported by
  \emph{Funda\c{c}\~{a}o para a Ci\^{e}ncia e a Tecnologia} and EU FEDER
  POCTI/POCI via SQIG - Instituto de Telecomunica\c{c}\~{o}es through
  the FCT project UID/EEA/50008/2013.}
\fntext[dga]{Olivier Bournez and Amaury Pouly were partially supported by
  \emph{DGA Project CALCULS}}

\begin{abstract}
In 1941, Claude Shannon introduced the General Purpose Analog
Computer (GPAC) as a mathematical model of Differential Analysers,
that is to say as a model of continuous-time analog (mechanical, and
later on electronic) machines of that time.

Following Shannon's arguments, functions generated by the GPAC must satisfy
a polynomial differential algebraic equation (DAE). As it is known that some computable functions
like Euler's $\Gamma(x)=\int_{0}^{\infty}t^{x-1}e^{-t}dt$ or Riemann's Zeta function
$\zeta(x)=\sum_{k=0}^\infty \frac1{k^x}$ do not satisfy any polynomial DAE,
this argument has often been used to demonstrate in the past
that the GPAC is less powerful than digital computation.

It was proved in \citep{JOC2007}, that if a more modern notion of
computation is considered, i.e. in particular if computability is not
restricted to real-time generation of functions, the GPAC is actually
equivalent to Turing machines.

Our purpose is first to discuss the robustness of the  notion
of computation involved in \citep{JOC2007}, by establishing that many
natural variants of the notion of
computation from this paper lead to the same computability result. 

Second, to go from these computability results towards considerations
about (time) complexity: we explore several natural variants for
measuring time/space complexity of a computation.

Quite surprisingly, whereas defining a robust time complexity for
general continuous time systems is a well known open problem, we prove
that all variants are actually equivalent even at the complexity
level. As a consequence, it seems that a robust and well defined
notion of time complexity exists for the GPAC, or equivalently for
computations by polynomial ordinary differential equations. 

Another side effect of our proof is also that we show in some way that
polynomial ordinary differential equations can actually be used as a
kind of programming model, and that there is a rather nice and robust
notion of ordinary differential equation (ODE) programming.
\end{abstract}


\maketitle

\section{Introduction}

Claude Shannon introduced in \citep{Sha41} the General Purpose Analog
Computer (GPAC) as a model for Differential Analysers
\citep{Bus31}, which are mechanical (and later on electronic)
continuous time analog machines, on which he worked as an
operator. The model was later refined in \citep{Pou74}, \citep{GC03}.
It was originally presented by Shannon as a model based on
circuits. Basically, a GPAC is any circuit (loops are allowed\footnote{There are some syntactic restrictions
to avoid ill-defined circuits.}) that can be built from the
4 basic units of \figref{fig:gpac_circuit}, which implement
constants, addition, multiplication and
integration, all of them working over analog real quantities (that
were corresponding to angles in the mechanical Differential Analysers,
and later on to voltages in the electronic versions). Note that the set of allowed
constants will generally be restricted, for example to rational numbers, to avoid pathological issues.
Given such a circuit, the function which gives the value of every wire (or
a subset of the wires) over time is said to be \emph{generated} by the circuit.
In \defref{def:gpac_generable_ext}, we consider an extension of this notion.

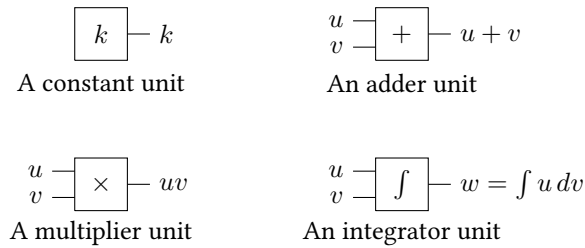
\begin{figure}[h]
\begin{center}
 \setlength{\unitlength}{1200sp}%
\begin{tikzpicture}
 \begin{scope}[shift={(0,0)},rotate=0]
  \draw (0,0) -- (0.7,0) -- (0.7,0.7) -- (0,0.7) -- (0,0);
  \node at (.35,.35) {$k$};
  \draw (.7,.35) -- (1,.35);
  \node[anchor=west] at (1,.35) {$k$};
  \node at (.35, -.3) {A constant unit};
 \end{scope}
 \begin{scope}[shift={(4,0)},rotate=0]
  \draw (0,0) -- (0.7,0) -- (0.7,0.7) -- (0,0.7) -- (0,0);
  \node at (.35,.35) {$+$};
  \draw (.7,.35) -- (1,.35); \draw (-.3,.175) -- (0,.175); \draw (-.3,.525) -- (0,.525);
  \node[anchor=west] at (1,.35) {$u+v$};
  \node at (.35, -.3) {An adder unit};
  \node[anchor=east] at (-.3,.525) {$u$};
  \node[anchor=east] at (-.3,.175) {$v$};
 \end{scope}
 \begin{scope}[shift={(0,-2)},rotate=0]
  \draw (0,0) -- (0.7,0) -- (0.7,0.7) -- (0,0.7) -- (0,0);
  \node at (.35,.35) {$\times$};
  \draw (.7,.35) -- (1,.35); \draw (-.3,.175) -- (0,.175); \draw (-.3,.525) -- (0,.525);
  \node[anchor=west] at (1,.35) {$uv$};
  \node at (.35, -.3) {A multiplier unit};
  \node[anchor=east] at (-.3,.525) {$u$};
  \node[anchor=east] at (-.3,.175) {$v$};
 \end{scope}
 \begin{scope}[shift={(4,-2)},rotate=0]
  \draw (0,0) -- (0.7,0) -- (0.7,0.7) -- (0,0.7) -- (0,0);
  \node at (.35,.35) {$\int$};
  \draw (.7,.35) -- (1,.35); \draw (-.3,.175) -- (0,.175); \draw (-.3,.525) -- (0,.525);
  \node[anchor=west] at (1,.35) {$w=\int u\thinspace dv$};
  \node at (.35, -.3) {An integrator unit};
  \node[anchor=east] at (-.3,.525) {$u$};
  \node[anchor=east] at (-.3,.175) {$v$};
 \end{scope}
\end{tikzpicture}
\end{center}
\caption{Circuit presentation of the GPAC: a circuit built from basic
  units. Presentation of the 4 types of units: constant, adder,
  multiplier, and integrator.}
\label{fig:gpac_circuit}
\end{figure}%

An important aspect of this model is that despite the apparent simplicity of its
basic blocks, sophisticated functions can easily be generated. \figref{fig:gpac_example_sin} illustrates how the
sine function can be generated using two integrators, with suitable
initial states. Incidentally, the sine function is also the solution of a very simple
 ordinary differential equation.
Shannon itself realized that functions generated by a GPAC
are nothing more than solutions of a special class of polynomial
differential equations. In particular it can be shown that a function
$f: \R \to \R$ is generated by Shannon's model \citep{Sha41},
\citep{GC03} if and only if it is a (component of the) solution of a
polynomial initial value problem (PIVP) of the form:
\begin{equation}\label{eq:ode}
\left\{\begin{array}{@{}r@{}l}y'(t)&=p(y(t))\\y(t_0)&=y_0\end{array}\right.,\qquad t\in\R
\end{equation}
where $p$ is a vector of polynomials and $y(t)$ is vector. In other
words, $f(t)=y_1(t)$, and $y_i'(t)=p_i(y(t))$
where $p_i$ is a multivariate polynomial.

\begin{figure}
\begin{center}
\begin{tikzpicture}
 \begin{scope}[shift={(-4.2,0)},rotate=0]
  \draw (0,0) -- (0.7,0) -- (0.7,0.7) -- (0,0.7) -- (0,0);
  \node at (.35,.35) {$-1$};
 \end{scope}
 \draw (-3.5,.35) -- (-3.15,.35) -- (-3.15,.175) -- (-2.8,.175);
 \begin{scope}[shift={(-2.8,0)},rotate=0]
  \draw (0,0) -- (0.7,0) -- (0.7,0.7) -- (0,0.7) -- (0,0);
  \node at (.35,.35) {$\times$};
 \end{scope}
 \draw (-2.1,.35) -- (-1.75,.35) -- (-1.75,.525) -- (-1.4,.525);
 \begin{scope}[shift={(-1.4,0)},rotate=0]
  \draw (0,0) -- (0.7,0) -- (0.7,0.7) -- (0,0.7) -- (0,0);
  \node at (.35,.35) {$\int$};
 \end{scope}
 \draw (-.7,.35) -- (-.35,.35) -- (-.35,.525) -- (0,.525);
 \begin{scope}[shift={(0,0)},rotate=0]
  \draw (0,0) -- (0.7,0) -- (0.7,0.7) -- (0,0.7) -- (0,0);
  \node at (.35,.35) {$\int$};
 \end{scope}
 \draw (.7,.35) -- (1.4,.35);
 \node[anchor=west] at (1.4,.35) {$\sin(t)$};
 \node[anchor=north] at (-1, -0.5) {$\left\lbrace
\begin{array}{@{}c@{}l}
y'(t)&=z(t)\\
z'(t)&=-y(t)\\
y(0)&=0\\
z(0)&=1
\end{array}
\right.\Rightarrow\left\lbrace\begin{array}{@{}c@{}l}
y(t)&=\sin(t)\\
z(t)&=\cos(t)
\end{array}\right.$};
 \draw (1,.35) -- (1,1) -- (-3.15,1) -- (-3.15,.525) -- (-2.8,.525);
 \fill (1,.35) circle[radius=.07];
 \draw (-4.9,.35) -- (-4.55,.35) -- (-4.55,-.3) -- (-.3,-.3) -- (-.3,.175) -- (0,.175);
 \draw (-1.75,-.3) -- (-1.75,.175) -- (-1.4,.175);
 \fill (-1.75,-.3) circle[radius=.07];
 \node[anchor=east] at (-4.9,.35) {$t$};
\end{tikzpicture}
\end{center}
\caption{Example of a GPAC circuit computing the sine and cosine.}
\label{fig:gpac_example_sin}
\end{figure}
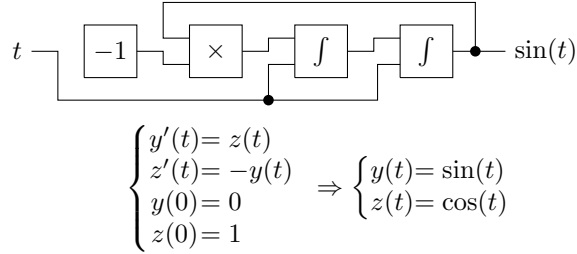

Intuitively, the link between a GPAC and a PIVP is the following: the idea is just to introduce a variable for each output of
a basic unit, and write the corresponding ordinary differential
equation (ODE), and observe that it can be written as an ODE with a polynomial right hand side.

While many of the usual real functions are known to be generated by a
GPAC, a notable exception is
Euler's Gamma function $\Gamma(x)=\int_{0}%
^{\infty}t^{x-1}e^{-t}dt$ function or Riemann's Zeta function
$\zeta(x)=\sum_{k=0}^\infty \frac1{k^x}$
\citep{Sha41}, \citep{PR89}, which are known not to satisfy any polynomial DAE,
i.e. they are not solutions of a system of the form \eqref{eq:ode}. If we have in mind that these functions are
known to be computable under the computable analysis framework
\citep{PR89}, \citep{Wei00} the previous result has long been interpreted as evidence
that the GPAC is a somewhat weaker model than computable analysis.

In 2007, it was proved that this is more an artifact of the notion of
real-time generation considered by Shannon than a true consideration
about the computational power of the model. Indeed, Shannon assumes
the GPAC computes in \emph{``real time''} - a very restrictive form of
computation: at time $t$ the output of the machine must be $\Gamma(t)$. If we change
this notion of computability to the kind of \emph{``converging computation''}
used in recursive analysis, or in modern computability theory, then
the $\Gamma$ function becomes computable \citep{Gra04}, and more
generally all functions over a bounded domain, computable in the sense
of computable analysis, are actually GPAC computable (and conversely)
\citep{JOC2007}. The idea used in \citep{Gra04}, \citep{JOC2007} to compute a function $f:\R\rightarrow\R$ is to define a polynomial
 initial-value problem (PIVP) \eqref{eq:ode} such that the argument $x$ of $f$ is provided to the
 PIVP via the initial condition, and the system has a component which converges to $f(x)$.
 Moreover, the convergence rate of the component to $f(x)$ is known and we know exactly how
 much time we have to wait to get a desired accuracy when computing $f(x)$. More precisely,
 the following was proved:

\begin{definition}[GPAC computable function] \label{def:joc2007}
$f:\R\rightarrow\R$ is called \emph{GPAC-computable} if there are
polynomials $p$ and $q$ with computable coefficients
such that for any $x\in\R$, there exists (a unique) $y:I\rightarrow\R^d$ satisfying
for all $t\in\Rp$:
\begin{itemize}
\item $y(0)=q(x)$ and $y'(t)=p(y(t))$ \hfill$\blacktriangleright$ $y$ satisfies a PIVP
\item if $t\geqslant1$ then $|y_1(t)-f(x)|\leqslant e^{-t}$\hfill$\blacktriangleright$ $y_1$ converges to $f(x)$
\end{itemize}
\end{definition}

\begin{proposition}[\citep{JOC2007}] \label{prop:avant}
Let $a$ and $b$ be some computable reals. A function $f:[a,b] \to \R$
is computable\footnote{In the classical sense, i.e. in the sense of
  computable analysis.} if and only if it is GPAC-computable.
\end{proposition}

In this paper our purpose is twofold: first explore natural  variations on the notion of
computability presented in \defref{def:joc2007} and, second, go
towards complexity theory and not only computability theory, by
introducing some natural ways to measure complexity.

It is important to understand that talking about time complexity for
continuous-time systems is known to be a non-trivial issue. Indeed, defining a
robust (time) complexity notion for continuous time systems is a well
known open problem \citep{CIEChapter2007} with no generic solution
provided at this day. In short, the difficulty is that the naive idea
of using the time variable of the ODE as a measure of ``time
complexity'' is problematic, since time can be arbitrarily contracted
in a continuous system due to the ``Zeno phenomena'' (e.g.~by using
functions like $\arctan$ which contract the whole real line into a
bounded set). It follows that all computable languages can then be
computed by a continuous system in time $O(1)$ (see e.g.~
\citep{Ruo93}, \citep{Ruo94}, \citep{Moo95b}, \citep{Bournez97},
\citep{Bou99}, \citep{AD91}, \citep{CP01}, \citep{Davies01},
\citep{Copeland98}, \citep{Cop02}).

Two first natural quantities will be considered: first, the time variable of the ordinary
differential equation, that we will sometimes call \emph{time}, and a
bound on the norm of the involved variables, that we will sometimes
call \emph{space}. 

As a reparameterization of the time variable of an ordinary
differential equation leads to a new ordinary differential equation
with the same solution curve, but which is traveled along time at a different speed, a natural idea is
to try to consider quantities that are kept invariant by
reparameterization. A natural choice for such quantity is the length of the
curve. We recall that the length of a curve $y\in C^1(I,\R^n)$ defined
over some interval $I=[a,b]$ is given by
$\glen{y}(a,b)=\int_I\infnorm{y'(t)}dt.$
\defref{def:joc2007} leads then naturally to consider the following
natural two variants of computability of functions over $\R^n$ given below.

Given $x\in\R^n$, we write $x_i$ for the $i^{th}$
component of $x$ and $x_{i..j}$ for the vector $(x_i,x_{i+},\ldots,x_j)$.
$\Rpoly$ denotes the set of polynomial-time computable reals \citep{Wei00}.
$\K[\R^n]$ denotes polynomial functions with $n$ variables and with coefficients in $\K$,
where variables live in $\R^n$ and $\Rp=[0,+\infty[$. In this document, $f:\subseteq X\rightarrow Y$
denotes a partial function, \emph{i.e.} $f:Z\rightarrow Y$ where $X\subseteq Z$.
We also take $\pastsup{\delta} f(t)=\sup_{u\in[t,t-\delta]\cap\R_+}f(t)$. The
intuition is that in \defref{def:joc2007} we can reparameterize the time variable, but this
will happen at the cost of space.
Hence, if we want to know how many resources are needed
to compute $f(x)$ with some accuracy $\mu$, we should measure not only the time but also the space needed to obtain this accuracy.
This is done in \defref{def:gc}, while \defref{def:glc} is a variant which, instead of using of measuring accuracy against time and space,
does this by measuring accuracy against the length of the solution curve needed to achieve that accuracy.
Figures \ref{fig:glc} and \ref{fig:gc} illustrate those definitions.

\begin{remark}[The space $\K$ of the coefficients]
In this paper, the coefficients of all considered polynomials will belong to $\K$.
Formally, $\K$ needs to a be \emph{generable field}, as introduced in \citep{BGP16Arxiv}.
However, without a significant loss of generality, the reader can consider that
\[\K=\Rpoly\]
which is the set of polynomial time computable real numbers. All the reader needs to
know about $\K$ is that it is a field and it is stable by generable functions
(introduced in \secref{sec:shannon}), meaning that if $\alpha\in\K$ and $f$
is generable then $f(\alpha)\in\K$.
It is shown in \citep{BGP16Arxiv} that there exists a small generable field $\Rgen$
lying somewhere between $\Q$ and $\Rpoly$, with probable strict inequality on both sides.
\end{remark}

We now get to our first and main notion of computable function:

\begin{definition}[Analog Length Computability]\label{def:glc}
Let $n,m\in\N$, $f:\subseteq\R^n\rightarrow\R^m$ and $\Omega:\Rp^2\rightarrow\Rp$.
We say that $f$ is $\Omega$-length-computable if and only if there exist $d\in\N$,
and $p\in\K^d[\R^d],q\in\K^d[\R^n]$
such that for any $x\in\dom{f}$, there exists (a unique) $y:\Rp\rightarrow\R^d$ satisfying
for all $t\in\Rp$:
\begin{itemize}
\item $y(0)=q(x)$ and $y'(t)=p(y(t))$\hfill$\blacktriangleright$ $y$ satisfies a PIVP
\item for any $\mu\in\Rp$, if $\glen{y}(0,t)\geqslant\Omega(\infnorm{x},\mu)$
    then $\infnorm{y_{1..m}(t)-f(x)}\leqslant e^{-\mu}$\\\hphantom{a}\hfill $\blacktriangleright$ $y_{1..m}$ converges to $f(x)$
\item $\infnorm{y'(t)}\geqslant1$\hfill$\blacktriangleright$ technical condition: the
  length grows at least linearly with time\footnote{This is a
    technical condition required for the proof. This can be weakened,
    for example to $\infnorm{p(y(t))}\geqslant\frac{1}{\poly(t)}$. The
    technical issue is that if the speed of the system
becomes extremely small, it might take an exponential time to reach a
polynomial length, and we want to avoid such ``unnatural''
cases.}
\end{itemize}
We denote by $\glc{\Omega}$ the set of $\Omega$-length-computable
functions, and by $\gplc$ the set of
$(\poly)$-length-computable functions, and more generally by $\cglc$ the
length-computable functions (for some $\Omega$). 
\end{definition}

\begin{definition}[Analog Time-Space computability]\label{def:gc}
Let $n,m\in\N$, $f:\subseteq\R^n\rightarrow\R^m$ and $\Upsilon,\Omega:\Rp^2\rightarrow\Rp$.
We say that $f$ is $(\Upsilon,\Omega)$-time-space-computable if and only if there exist $d\in\N$,
and $p\in\K^d[\R^d],q\in\K^d[\R^n]$
such that for any $x\in\dom{f}$, there exists (a unique) $y:\Rp\rightarrow\R^d$ satisfying
for all $t\in\Rp$:
\begin{itemize}
\item $y(0)=q(x)$ and $y'(t)=p(y(t))$\hfill$\blacktriangleright$ $y$ satisfies a PIVP
\item for all $\mu\in\Rp$, if $t\geqslant\Omega(\infnorm{x},\mu)$ then $\infnorm{y_{1..m}(t)-f(x)}\leqslant e^{-\mu}$
\\\hphantom{a}\hfill$\blacktriangleright$ $y_{1..m}$ converges to $f(x)$
\item $\infnorm{y(t)}\leqslant\Upsilon(\infnorm{x},t)$, for all $t\geqslant0$\hfill$\blacktriangleright$ $y(t)$ is bounded
\end{itemize}
We denote by $\gc{\Upsilon}{\Omega}$ the set of
$(\Upsilon,\Omega)$-time-space-computable functions, by $\gpc$ the set of
$(\poly,\poly)$-time-space-computable functions, and by $\cgc$ the set
of time-space-computable functions.
\end{definition}

\begin{figure}
\centering
\begin{tikzpicture}[domain=1:11,samples=500,scale=1]
    \draw[very thin,color=gray] (0.9,-0.1) grid (11.1,4.1);
    \draw[->] (1,-0.1) -- (1,4.2);
    \draw[->] (0.9,0) -- (11.2,0) node[right] {$\glen{y}$};
    \draw[color=red,thick] plot[id=fn_1] function{
        1+exp((x-1)*(7-x)/10)+(1+sin(10*x))/(1+exp(x-3))-2*exp(-(x-1))
        };
    \draw[blue] (0.9,1) -- (11.1,1);
    \draw[right,blue] (11,1) node {$f(\textcolor{blue}{x})$};
    \draw[color=red,thick] (1.0,0.4) -- (0.9,0.4);
    \draw[red] (0.95,0.4) node[left] {$q_1(\textcolor{blue}{x})$};
    \draw[color=red] (3.3,2.7) node {$y_1$};
    \draw[very thick] (7.0,0) -- (7.0,-0.2);
    \draw[<->,mygreen,thick] (7.0,1) -- (7,2) node[midway,black,left] {$e^{-\textcolor{mygreen}{0}}$};
    \draw (7,-0.1) node[below] {$\Omega(\textcolor{blue}{x},\textcolor{mygreen}{0})$};
    \draw[<->,mygreen,thick] (8.4,1) -- (8.4,1.36) node[midway,black,left] {$\scriptstyle e^{-\textcolor{mygreen}{1}}$};
    \draw[dotted] (8.4,0) -- (8.4,1);
    \draw[very thick] (8.4,0) -- (8.4,-0.2);
    \draw (8.4,-0.1) node[below] {$\Omega(\textcolor{blue}{x},\textcolor{mygreen}{1})$};
\end{tikzpicture}
\caption{$\glc{\Omega}$: on input $x$, starting from initial condition $q(x)$,
the PIVP $y'=p(y)$ ensures that $y_1(t)$ gives $f(x)$ with accuracy better than $e^{-\mu}$
as soon as the length of $y$ (from $0$ to $t$) is greater than $\Omega(\infnorm{x},\mu)$.
Note that we did not plot the other variables $y_2,\ldots,y_d$ and the horizontal
axis measures the length of $y$ (instead of the time $t$).\label{fig:glc}}
\end{figure}
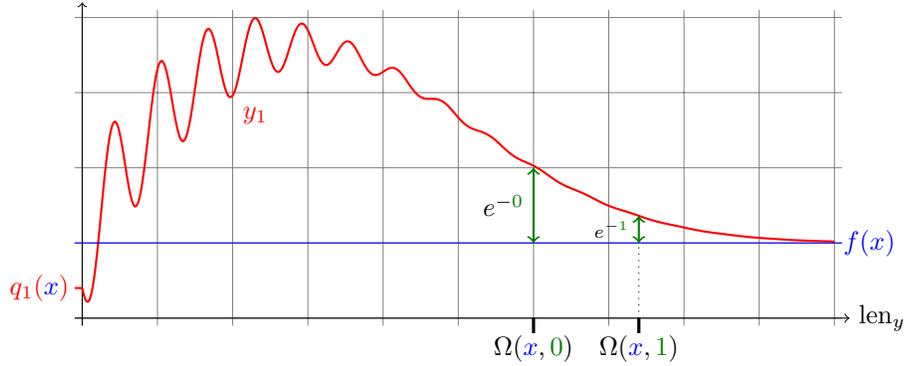

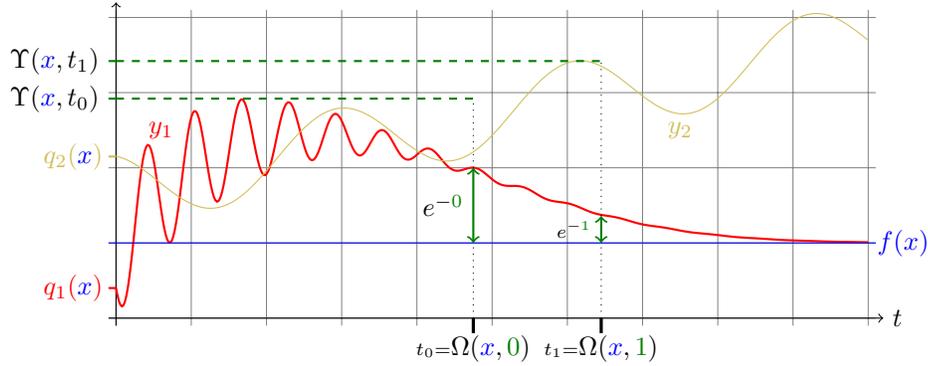
\begin{figure}
\centering
\begin{tikzpicture}[domain=1:11,samples=500,scale=1]
    \draw[very thin,color=gray] (0.9,-0.1) grid (11.1,4.1);
    \draw[->] (1,-0.1) -- (1,4.2);
    \draw[->] (0.9,0) -- (11.2,0) node[right] {$t$};
    \draw[color=red,thick] plot[id=fn_2] function{
        1+exp((x-1)*(7-x)/10)/2+(1+sin(10*x))/(1+exp(x-3))-1.5*exp(-(x-1))
        };
    \draw[blue] (0.9,1) -- (11.1,1);
    \draw[right,blue] (11,1) node {$f(\textcolor{blue}{x})$};
    \draw[color=red,thick] (1.0,0.4) -- (0.9,0.4);
    \draw[red] (0.95,0.4) node[left] {$q_1(\textcolor{blue}{x})$};
    \draw[color=red] (1.6,2.5) node {$y_1$};
    \draw[very thick] (5.75,0) -- (5.75,-0.2);
    \draw[<->,mygreen,thick] (5.75,1) -- (5.75,2) node[midway,black,left] {$e^{-\textcolor{mygreen}{0}}$};
    \draw (5.75,-0.1) node[below] {${\scriptstyle t_0=}\Omega(\textcolor{blue}{x},\textcolor{mygreen}{0})$};
    \draw[dotted] (5.75,0) -- (5.75,1);
    \draw[<->,mygreen,thick] (7.45,1) -- (7.45,1.36) node[midway,black,left] {$\scriptstyle e^{-\textcolor{mygreen}{1}}$};
    \draw[dotted] (7.45,0) -- (7.45,1);
    \draw[very thick] (7.45,0) -- (7.45,-0.2);
    \draw (7.45,-0.1) node[below] {${\scriptstyle t_1=}\Omega(\textcolor{blue}{x},\textcolor{mygreen}{1})$};
    \draw[color=myyellow] plot[id=fn_3] function {1.5+x/5+sin(x*2)/2};
    \draw[color=myyellow,thick] (1.0,2.15) -- (0.9,2.15);
    \draw[myyellow] (0.95,2.15) node[left] {$q_2(\textcolor{blue}{x})$};
    \draw[color=myyellow] (8.5,2.5) node {$y_2$};
    \draw[dashed,thick,mygreen] (1,2.92) -- (5.75,2.92);
    \draw[dotted] (5.75,2) -- (5.75,2.92);
    \draw[mygreen,thick] (0.9,2.92) -- (1,2.92);
    \draw (.9,2.92) node[left] {$\Upsilon(\textcolor{blue}{x},t_0)$};
    \draw[dashed,thick,mygreen] (1,3.42) -- (7.45,3.42);
    \draw[dotted] (7.45,1.36) -- (7.45,3.42);
    \draw[mygreen,thick] (0.9,3.42) -- (1,3.42);
    \draw (.9,3.42) node[left] {$\Upsilon(\textcolor{blue}{x},t_1)$};
\end{tikzpicture}
\caption{$\gc{\Upsilon}{\Omega}$: on input $x$, starting from initial condition $q(x)$,
the PIVP $y'=p(y)$ ensures that $y_1(t)$ gives $f(x)$ with accuracy better than $e^{-\mu}$
as soon as the time $t$ is greater than $\Omega(\infnorm{x},\mu)$. At the same time,
all variables $y_j$ are bounded by $\Upsilon(\infnorm{x},t)$. Note that variables $y_2,\ldots,y_d$
need not converge to anything.\label{fig:gc}}
\end{figure}

Indeed, \propref{prop:avant} can be reformulated as:

\begin{proposition}
Let $a$ and $b$ be some computable reals. A function $f:[a,b] \to \R$
is computable\footnote{In the classical sense, i.e. in the sense of
  computable analysis.} if and only if it is length-computable if and
only if it is time-space-computable. 
\end{proposition}

More surprisingly, we prove that it turns out that both classes are the same,
even at the complexity level.

\begin{theorem}\label{th:main-weak}
$\gplc=\gpc$.
\end{theorem}

This turns out suprisingly to also be equivalent with many variants, both at the
computability and complexity level.

For example, the error could also be given as input, via an initial condition.
The intuition behind the following definition is that the initial condition also depends on the accuracy $\mu$.
Hence, instead of what happens in \defref{def:gc}, we are not guaranteed that a component converges to $f(x)$,
only that it stays in a $e^{-\mu}$-vicinity of $f(x)$ after some time, and that the space used is bounded.

\begin{definition}[Analog weak computability]\label{def:gwc}
Let $n,m\in\N$, $f:\subseteq\R^n\rightarrow\R^m$, $\Omega:\Rp^2\rightarrow\Rp$
and $\Upsilon:\Rp^3\rightarrow\Rp$.
We say that $f$ is $(\Upsilon,\Omega)$-weakly-computable if and only if there exist $d\in\N$,
$p\in\K^d[\R^d],q\in\K^d[\R^{n+1}]$
such that for any $x\in\dom{f}$ and $\mu\in\Rp$, there exists (a unique) $y:\Rp\rightarrow\R^d$ satisfying
for all $t\in\Rp$:
\begin{itemize}
\item $y(0)=q(x,\mu)$ and $y'(t)=p(y(t))$\hfill$\blacktriangleright$ $y$ satisfies a PIVP
\item if $t\geqslant\Omega(\infnorm{x},\mu)$ then $\infnorm{y_{1..m}(t)-f(x)}\leqslant e^{-\mu}$
\hfill$\blacktriangleright$ $y_{1..m}$ approximates $f(x)$
\item $\infnorm{y(t)}\leqslant\Upsilon(\infnorm{x},\mu,t)$
\hfill$\blacktriangleright$ $y(t)$ is bounded
\end{itemize}
We denote by $\gwc{\Upsilon}{\Omega}$ the set of
$(\Upsilon,\Omega)$-weakly-computable functions, by $\gpwc$ the set of
$(\poly,\poly)$-weakly-computable functions, and by $\cgwc$ the set of
weakly-computable functions.
\end{definition}

Or we could consider a notion of online-computation, the intuition behind it being that if
some external input $x(t)$ approaches a value $\bar{x}$ sufficiently close, then by waiting
 enough time, and assuming that the external input stays near the value $\bar{x}$ during that time interval, we will get
an approximation of $f(\bar{x})$ with some desired accuracy. This process is illustrated in \figref{fig:goc}.
By constantly changing the external input $x(t)$ and ``locking it'' during some time near some value, we are able to compute approximations
of $f(x)$ for several arguments in a single ``run'' of the GPAC.

\begin{figure}
\centering
\begin{tikzpicture}[domain=1:11,samples=500,scale=1]
    \begin{scope}[shift={(0,5.5)}]
    \draw[very thin,color=gray] (0.9,-0.1) grid (11.1,3.1);
    \draw[->] (1,-0.1) -- (1,3.2);
    \draw[->] (0.9,0) -- (11.2,0) node[right] {$t$};
    \draw[color=red,thick] plot[domain=1:4,id=fn_4] function{2.25+sin(x*7*pi)/4*(3+tanh((x-3)*(x-3)))/2};
    \draw[blue] (0.9,2.25) -- (4.8,2.25);
    \draw[blue] (0.9,2.25) node[left] {$\bar{x}$};
    \draw[dashed,mygreen,thick] (1,2.75) -- (5.2,2.75);
    \draw[dashed,mygreen,thick] (1,1.75) -- (5.2,1.75);
    \draw[<->,mygreen,thick] (5,1.75) -- (5,2.75);
    \draw[mygreen] (4.95,2.25) -- (5.05,2.25);
    \draw (5,2.45) node[right] {$e^{-\Lambda(\textcolor{blue}{\bar{x}},\textcolor{mygreen}{1})}$};

    \draw[red,thick] (4,2.25) -- (7,.5);

    \draw[color=red,thick] plot[domain=7:11,id=fn_5] function{0.5+sin(x*3*(3+tanh((x-9)*(x-9)))/4*pi)/4*(3+tanh((x-10)*(x-10)))/4};
    \draw[dashed,mygreen,thick] (6.2,0.75) -- (11,0.75);
    \draw[dashed,mygreen,thick] (6.2,0.25) -- (11,0.25);
    \draw[<->,mygreen,thick] (6.4,0.25) -- (6.4,0.75);
    \draw[mygreen] (6.35,.5) -- (6.45,.5);
    \draw (6.4,.5) node[left] {$e^{-\Lambda(\textcolor{blue}{\bar{x}'},\textcolor{mygreen}{2})}$};
    \draw[blue] (6.6,.5) -- (11.1,.5);
    \draw[blue] (11.1,.5) node[right] {$\bar{x}'$};
    \end{scope}

    \begin{scope}[shift={(0,4.75)}]
    \draw[thick,myyellow,|-|] (1,0) -- (3.2,0);
    \draw[thick,myyellow,-|] (3.2,0) -- (4.87,0);
    \draw[thick,myyellow,|-|] (6.54,0) -- (9.2,0);
    \draw[thick,myyellow,-|] (9.2,0) -- (11,0);
    \end{scope}

    \draw[dotted,myyellow,thick] (4.85,0) -- (4.85,7.25);
    \draw[dotted,myyellow,thick] (3.18,0) -- (3.18,4.75);
    \draw (2.1,4.7) node[below,myyellow] {undefined};
    \draw (4,4.7) node[below,myyellow] {accurate};
    \draw (3,4.8) node[above,myyellow] {stable};

    \draw (5.7,4.8) node[above,myyellow] {unstable};
    \draw (5.7,4.7) node[below,myyellow] {undefined};
    
    \draw[dotted,myyellow,thick] (6.55,0) -- (6.55,6.25);
    \draw[dotted,myyellow,thick] (9.18,0) -- (9.18,4.75);
    \draw (7.9,4.7) node[below,myyellow] {undefined};
    \draw (10.1,4.7) node[below,myyellow] {accurate};
    \draw (9.1,4.8) node[above,myyellow] {stable};

    \draw[very thin,color=gray] (0.9,-0.1) grid (11.1,4.1);
    \draw[->] (1,-0.1) -- (1,4.2);
    \draw[->] (0.9,0) -- (11.2,0) node[right] {$t$};
    \draw[color=red,thick] plot[id=fn_6] function{
        1+exp((x-1)*(7-x)/10)/2+(1+sin(10*x))/(1+exp(x-2.8))-1.5*exp(-(x-1))
        +(1+tanh(100*(x-6)))/2*sin(10*x)*exp((6-x)/1.75)
        };
    \draw[blue] (0.9,2.42) -- (4.85,2.42);
    \draw[right,blue] (0.95,2.42) node[left] {$f(\textcolor{blue}{\bar{x}})$};
    \draw[color=red,thick] (1.0,0.4) -- (0.9,0.4);
    \draw[red] (0.95,0.4) node[left] {$y_0$};
    \draw[color=red] (2.6,1.5) node {$y_1$};
    \draw[mygreen,dashed,thick] (3.2,2.8) -- (6,2.8);
    \draw[mygreen,dashed,thick] (3.2,2.08) -- (6,2.08);
    \draw[<->,mygreen,thick] (5.75,2.08) -- (5.75,2.8) node[midway,black,left] {$e^{-\textcolor{mygreen}{1}}$};
    \draw[mygreen] (5.7,2.44) -- (5.8,2.44);

    \draw[very thick] (1,0) -- (1,-0.2);
    \draw (1,-.1) node[below] {$t_1$};
    \draw[very thick] (3.2,0) -- (3.2,-0.2);
    \draw (3.2,-0.1) node[below] {${\scriptstyle t_1+}\Omega(\textcolor{blue}{\bar{x}},\textcolor{mygreen}{1})$};

    \draw[blue] (9.2,1.05) -- (11.1,1.05);
    \draw[right,blue] (11,1.05) node[right] {$f(\textcolor{blue}{\bar{x}'})$};
    \draw[very thick] (6.55,0) -- (6.55,-0.2);
    \draw (6.55,-.1) node[below] {$t_2$};
    \draw[very thick] (9.2,0) -- (9.2,-0.2);
    \draw (9.2,-0.1) node[below] {${\scriptstyle t_2+}\Omega(\textcolor{blue}{\bar{x}'},\textcolor{mygreen}{2})$};
    \draw[mygreen,dashed,thick] (9.2,1.19) -- (11,1.19);
    \draw[mygreen,dashed,thick] (9.2,0.91) -- (11,0.91);
\end{tikzpicture}
\caption{$\goc{\Upsilon}{\Omega}{\Lambda}$: starting from the (constant) initial condition $y_0$,
the PIVP $y'(t)=p(y(t),x(t))$ has two possible behaviors depending on the input signal $x(t)$.
If $x(t)$ is unstable, the behaviour of the PIVP $y'(t)=p(y(t),x(t))$ is undefined.
If $x(t)$ is stable around $\bar{x}$ with error at most $e^{-\Lambda(\infnorm{\bar{x}},\mu)}$
then $y(t)$ is initially undefined, but after a delay of at most $\Omega(\infnorm{\bar{x}},\mu)$, $y_1(t)$ gives $f(\bar{x})$
with accuracy better than $e^{-\mu}$.
In all cases, all variables $y_j(t)$ are bounded by a function ($\Upsilon$)
of the time $t$ and the supremum of $\infnorm{x(u)}$ during a small time interval $u\in[t-\delta,t]$.
\label{fig:goc}}
\end{figure}
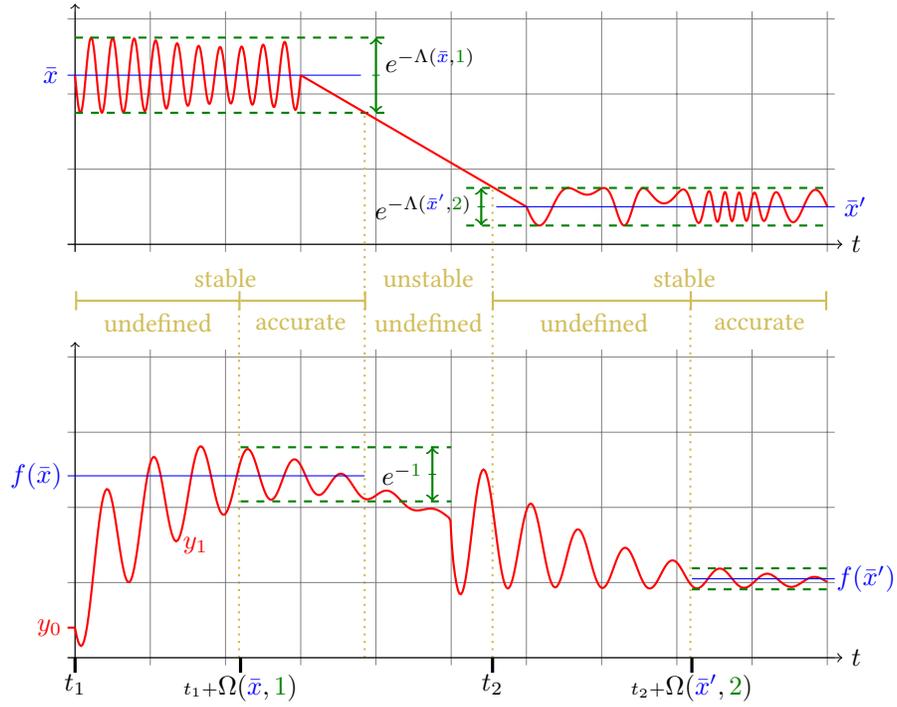

\begin{definition}[Online computability]\label{def:goc}
Let $n,m\in\N$, $f:\subseteq\R^n\rightarrow\R^m$ and $\Upsilon,\Omega,\Lambda:\Rp^2\rightarrow\Rp$.
We say that $f$ is $(\Upsilon,\Omega,\Lambda)$-online-computable if and only if there exist $\delta\geqslant0$, $d\in\N$ and
$p\in\K^d[\R^d\times\R^{n}]$ and $y_0\in\K^d$ such that for any $x\in C^0(\Rp,\R^n)$,
there exists (a unique) $y:\Rp\rightarrow\R^d$ satisfying for all $t\in\Rp$:
\begin{itemize}
\item $y(0)=y_0$ and $y'(t)=p(y(t),x(t))$
\item $\infnorm{y(t)}\leqslant\Upsilon\big(\pastsup{\delta}{\infnorm{x}}(t),t\big)$
\item for any $I=[a,b]\subseteq\Rp$, if there exist $\bar{x}\in\dom{f}$ and $\bar{\mu}\geqslant0$ such that
for all $t\in I$, $\infnorm{x(t)-\bar{x}}\leqslant e^{-\Lambda(\infnorm{\bar{x}},\bar{\mu})}$ then
$\infnorm{y_{1..m}(u)-f(\bar{x})}\leqslant e^{-\bar{\mu}}$ whenever
$a+\Omega(\infnorm{\bar{x}},\bar{\mu})\leqslant u\leqslant b$.
\end{itemize}
We denote by $\goc{\Upsilon}{\Omega}{\Lambda}$ the set of
$(\Upsilon,\Omega,\Lambda)$-online-computable, by $\gpoc$ the set of
$(\poly,\poly,\poly)$-online-computable functions and by $\cgoc$ the
set of online-computable functions.
\end{definition}

\begin{theorem} \label{th:main}
All notions of computations are equivalent, both at the computability
level:
$$\cglc=\cgc=\cgwc=\cgoc$$
and at the complexity level:
$$\gplc=\gpc=\gpwc=\gpoc$$
\end{theorem}

The rest of the current paper is devoted to prove these equivalences between
definitions.  In \secref{sec:shannon} we recall some results
established by \citep{Sha41}, and generalize several of them to
multivariate functions.  The proof of the previous \thref{th:main} then
follows but is however
rather involved, and requires the introduction of other equivalent
intermediate classes. 
We show several inclusions between these classes which will guarantee the result of \thref{th:main}. First we show that
$\gpc \subseteq \gpwc$, which follows from
the fact that it is possible to rescale the system using the length
of the curve as a new variable to make sure it does not grow faster
than a polynomial 
(\secref{app:gpcsubsetgpwc}). 
The other direction ($\gpwc \subseteq \gpc$) is really harder: 
the first
step is to transform a computation into a computation that tolerates
small perturbations of the dynamics ($\gpwc \subseteq \gprc$, Section
\ref{app:weak_to_robust}). The second problem is to avoid that the system
explodes for inputs not in the domain of the function
($\gprc \subseteq \gpsc$, Section
\ref{app:robust_to_strong}). As a third step, we allow the system to have its inputs (input and precision) changed
during the computation, but we require that the system has a maximum delay to react to these changes
($\gpsc \subseteq \gpuc$, \secref{app:strong_to_unaware}). 
Finally, as a fourth step, we add a mechanism that feeds the
system with the input and some precision. By continuously increasing
the precision with time, we ensure that the system will converge when
the input is stable. The result of these 4 steps is a
lemma yielding a
nice notion of online-computation ($\gpuc \subseteq \gpoc$,
\secref{app:unaware_to_online}). Equality
$\gpc=\gpwc=\gpoc$ follows because time and length are related for polynomially bounded systems.

A side effect of the closure properties of these classes, and of our
proofs,  is that 
programming with (polynomial length) ODE becomes a pleasant
exercise, once the logic is understood. For example, simulating the
assignment $y:=g_\infty$ corresponds to the dynamics of $y(0)=y_0$,
$y'(t)=\reach(y(t),g(t))+E(t)$, for a fixed function
$\reach$, tolerating bounded error $E(t)$ on dynamics, and $g$ fluctuating around $g_\infty$
.  Other
example: from a $\gpc$ system computing $f$, just adding the
corresponding $\gpoc$-equations for $g$, yields a \PIVP~computing $g \circ f$ 
, by feeding the output of the system computing $f$ to the (online)
input of $g$.

\section{The PIVP Class}
\label{sec:shannon}

This sections recalls some known results about the class of functions generated
by polynomial initial value problems. We omitted the proofs but this section contains
all the necessary definitions and theorems needed to make this paper self-contained. Other lemmas
related to the PIVP class are introduced in the paper when needed to avoid a long
list of lemmas. A much more complete and detailed analysis of this class, with
all the proofs, can be found in \citep{BGP16Arxiv} but we give a short overview below.

Terminology is important here:
the functions of this class are called \emph{generable}, and should not be confused
with the notion of computable function introduced earlier.
Informally, the main results on this class are the following:
\begin{itemize}
\item this class is stable by arithmetic operations and composition;
\item this class contains many useful functions such as trigonometric functions;
\item if $y'=f(y)$ where $f$ in this class, then $y$ is also in this class.
\end{itemize}
The general idea is that working directly with polynomial differential equations
is a perilous exercise but it becomes easier if we can use more than polynomials.
For example, assume that the above results are true, and consider the following differential equation:
\[y(0)=1, y'(t)=\sin(y(t)).\]
It can be seen that $\sin$ is generable so it follows that $y$ is generable.
Another example is the following differential equation:
\[y(0)=1, y'(t)=\tanh(y(t)^2).\]
It can be seen again that $\tanh$ is generable, and polynomials are also
generable so $x\mapsto\tanh(x^2)$ is generable, thus $y$ is generable.
Hopefully these two examples will convince the reader that this class gives us
a lot of flexibility when writing differential systems.

Another important aspect of this class is the growth of the functions. Without restrictions,
it is very easy to build fast-growing functions, such as towers of exponentials.
In this work, we crucially need to bound the growth of functions to limit the
power of our systems. A necessary condition for this is that we should only write differential
equations of the form $y'=f(y)$ where $f$ is generable and bounded by a polynomial.
Of course this condition is trivially satisfied by polynomials but is also
verified by many other functions such as $\sin$ or $\tanh$.

The following concept can be attributed to \citep{Sha41}:
a function $f: \R \to \R$ is said to be a \PIVP{} function if there
 exists a system of the form \eqref{eq:ode}  with $f(t)=y_1(t)$ for all $t$, where $y_1$
denotes the first component of the vector $y$ defined in $\R^d$. 
We need in our proof to extend this concept to talk about (i)
multivariable functions and (ii) the growth of these functions.  This
leads to the following:

\begin{definition}[Generable function \citep{BGP16Arxiv}]\label{def:gpac_generable_ext}
Let $d,e\in\N$, $I$ be an open and connected subset of $\R^d$, $\mtt{sp}:\R\rightarrow\Rp$
and $f:I\rightarrow\R^e$. We say that $f\in\gval[\K]{\mtt{sp}}$ if and only if
there exist $n\geqslant e$, $p\in\MAT{n}{d}{\K}[\R^n]$, $x_0\in\K^d$, $y_0\in\K^n$
and $y:I\rightarrow\R^n$ satisfying for all $x\in I$:
\begin{itemize}
\item $y(x_0)=y_0$ and $\jacobian{y}(x)=p(y(x))$ (i.e. $\partial_jy_i(x)=p_{ij}(y(x))$)\\
    \hphantom{a}\hfill$\blacktriangleright$ $y$ satisfies a differential equation
\item $f(x)=y_{1..e}(x)$\hfill$\blacktriangleright$ $f$ is a component of $y$
\item $\infnorm{y(x)}\leqslant \mtt{sp}(\infnorm{x})$\hfill$\blacktriangleright$ $y$ is bounded by $\mtt{sp}$
\end{itemize}
\end{definition}

\begin{definition}[Polynomially bounded generable function]
The class of generable functions with polynomially bounded value is called $\gpval$:
\[f\in\gpval\Leftrightarrow\text{ there exists a polynomial $\mtt{sp}$ such that }f\in\gval{\mtt{sp}}\]
\end{definition}

The following 
closure properties can be seen as extensions of the results from
\citep{GBC09} to multivariate functions:  

\begin{lemma}[Arithmetic on generable functions
  \citep{BGP16Arxiv}]\label{lem:gpac_ext_class_stable} 
Let $d$, $e$, $n$, $m\in\N$, $\mtt{sp},\ovl{\mtt{sp}}:\R\rightarrow\Rp$,
$f:\subseteq\R^d\rightarrow\R^n\in\gval{\mtt{sp}}$ and
$g:\subseteq\R^e\rightarrow\R^m\in\gval{\ovl{\mtt{sp}}}$. Then:
\begin{itemize}
\item $f+g, f-g\in\gval{\mtt{sp}+\ovl{\mtt{sp}}}$ over $\dom{f}\cap\dom{g}$ if $d=e$ and $n=m$
\item $fg\in\gval{\max(\mtt{sp},\ovl{\mtt{sp}},\mtt{sp}\thinspace\ovl{\mtt{sp}})}$ if $d=e$ and $n=m$
\item $f\circ g\in\gval{\max(\ovl{\mtt{sp}},\mtt{sp}\circ\ovl{\mtt{sp}})}$ if $m=d$ and $g(\dom{g})\subseteq \dom{f}$
\end{itemize}
\end{lemma}

Our key result is that the solution to an ODE whose right-hand side is
generable, and possibly depends on an external and $C^1$ control, may be rewritten
as a GPAC. A corollary of this result is that the solution of a generable ODE is
generable.

\begin{proposition}[Generable ODE rewriting \citep{BGP16Arxiv}]\label{prop:gpac_ext_ivp_stable_pre}
Let $d,n\in\N$, $I\subseteq\R^n$, $X\subseteq\R^d$, $\mtt{sp}:\Rp\rightarrow\Rp$ and
$(f:I\times X\rightarrow\R^n)\in\gval[\K]{\mtt{sp}}$. Define $\ovl{\mtt{sp}}=\max(\idfun,\mtt{sp})$.
Then there exist $m\in\N$, $(g:I\times X\rightarrow\R^m)\in\gval[\K]{\ovl{\mtt{sp}}}$
and $p\in\K^m[\R^m\times\R^d]$ such that for any interval $J$,
$t_0\in\K\cap J$, $y_0\in\K^n\cap J$, $y\in C^1(J,I)$ and $x\in C^1(J,X)$,
if $y$ satisfies:
\[\left\{\begin{array}{@{}r@{}l@{}}y(t_0)&=y_0\\y'(t)&=f(y(t),x(t))\end{array}\right.
\]
then there exists $z\in C^1(J,\R^m)$ such that:
\[\left\{\begin{array}{@{}r@{}l@{}}z(t_0)&=g(y_0,x(t_0))
\\z'(t)&=p(z(t),x'(t))\end{array}\right.
\qquad \left\{\begin{array}{@{}r@{}l@{}}y(t)&=z_{1..d}(t)\\\infnorm{z(t)}&\leqslant\ovl{\mtt{sp}}(\infnorm{y(t),x(t)})\end{array}\right.
\]
\end{proposition}

A simplified version of this lemma shows that generable functions are closed under
ODE solving.

\begin{corollary}[Closure under ODE of generable functions \citep{BGP16Arxiv}]\label{cor:gpac_ext_ivp_stable}
Let $d\in\N$, $J\subseteq\R$ an interval, $\mtt{sp},\ovl{\mtt{sp}}:\Rp\rightarrow\Rp$,
$f:\subseteq\R^d\rightarrow\R^d$ in \gval{\mtt{sp}}, $t_0\in\K\cap J$ and $y_0\in\K^d\cap\dom{f}$.
Assume there exists $y:J\rightarrow\dom{f}$ satisfying for all $t\in J$:
\[\left\{\begin{array}{@{}r@{}l@{}}y(t_0)&=y_0\\y'(t)&=f(y(t))\end{array}\right.
\qquad\infnorm{y(t)}\leqslant\ovl{\mtt{sp}}(t)\]
Then $y\in\gval{\max(\ovl{\mtt{sp}},\mtt{sp}\circ \ovl{\mtt{sp}})}$ and is unique.
\end{corollary}

It follows that many 
polynomially bounded usual analytic\footnote{Functions from $\gpval$ are necessarily
analytic, as solutions of an analytic ODE are analytic.} functions are in the class
$\gpval$.

We will also need the following results, which tell us how the solution of a GPAC varies if there is a slight change in the parameters defining it. In the next theorem $\sigmap{p}$ denotes the sum of the absolute values of the coefficients of the polynomial $p$.

\begin{theorem}[Parameter dependency 
  \citep{BGP16Arxiv}]\label{th:dependency_right_side}
Let $I=[a,b]$, $p\in\R^n[\R^{n+d}]$, $k=\degp{p}$,
$e\in C^0(I,\R^d)$, $x,\delta\in C^0(I,\R^n)$ and $y_0,z_0\in\R^d$.
Assume that $y,z:I\rightarrow\R^d$ satisfy:
\[\left\{\begin{array}{@{}r@{}l}y(a)&=y_0\\y'(t)&=p(y(t),x(t))\end{array}\right.\qquad
\left\{\begin{array}{@{}r@{}l}z(a)&=z_0\\z'(t)&=e(t)+p(z(t),x(t)+\delta(t))\end{array}\right.\qquad t\in I\]
Assume that there exists $\varepsilon>0$ such that for all $t\in I$,
\begin{multline}
\mu(t):= \\
\left(\infnorm{z_0-y_0}+\int_{a}^t\infnorm{e(u)}+k\sigmap{p}M^{k-1}(u)\infnorm{\delta(u)}du
\right) 
\exp \left(k\sigmap{p}\int_{a}^tM^{k-1}(u)du\right) \\ <\varepsilon
\end{multline}
where $M(t)=\varepsilon+\infnorm{y(t)}+\infnorm{x(t)}+\infnorm{\delta(t)}$. Then for all $t\in I$,
\[\infnorm{z(t)-y(t)}\leqslant\mu(t)\]
\end{theorem}

\begin{lemma}[Modulus of continuity \citep{BGP16Arxiv}]\label{prop:generable_mod_cont}
Let $\mtt{sp}:\Rp\rightarrow\Rp$, $f\in\gval{\mtt{sp}}$. There exists $q\in\K[\R]$ such that
for any $x_1,x_2\in\dom{f}$, if $[x_1,x_2]\subseteq\dom{f}$ then
\[\infnorm{f(x_1)-f(x_2)}\leqslant\infnorm{x_1-x_2}q(\mtt{sp}(\max(\infnorm{x_1},\infnorm{x_2}))).\]
In particular, if $f\in\gpval[\K]$ then there exists $q\in\K[\R]$ such that
if $[x_1,x_2]\subseteq\dom{f}$ then
\[\infnorm{f(x_1)-f(x_2)}\leqslant\infnorm{x_1-x_2}q(\max(\infnorm{x_1},\infnorm{x_2})).\]
\end{lemma}

After these statements, we can go to the proof of Theorem
\ref{th:main}. This is done by proving various implications.

\section{Proof that ALP is ATSP} 
\label{app:gpcsubsetgpwc}

The purpose of the current section is to show the following.

\begin{theorem}
$\gpc=\gplc$.
\end{theorem}

\subsection{Some remarks}

We start by a remark:

\begin{lemma}[Norm function, \citep{BGP16Arxiv}]\label{lem:norm}
There is a family of functions $\norm_{\infty,\delta}\in\gpval$ 
such that,
for any $x\in\R^n$ and $\delta\in]0,1]$, we have:
\[\infnorm{x}\leqslant\norm_{\infty,\delta}(x)\leqslant\infnorm{x}+\delta.\]
\end{lemma}

\subsection{The proof}

%

In one direction the proof is simple because if the system uses polynomial time and space
then there is a relationship between time and length and we only need to add one
variable to the system to make sure that the technical condition holds.
The other direction is more involved because we need to rescale the system using
the length of the curve to make sure it does not grow faster than a polynomial,
which is ensured by the technical condition.

Let $f\in\gc{\Upsilon}{\Omega}$ where $\Upsilon$ and $\Omega$ are polynomials, which
we assume to be increasing functions. Apply \defref{def:gc} to get $d,p,q$, let $k=\degp{p}$
and define:
\[\Omega^*(\alpha,\mu)=\Omega(\alpha,\mu)\left(1+\sigmap{p}\max\big(1,\Upsilon(\infnorm{x},\Omega(\alpha,\mu))\big)^k\right)\]
Let $x\in\dom{f}$ and consider the following system:
\[\left\{\begin{array}{@{}r@{}l}
y(0)&=q(x)\\z(0)&=0\end{array}\right.\qquad
\left\{\begin{array}{@{}r@{}l}
y'(t)&=p(y(t))\\z'(t)&=1\end{array}\right.
\]
Note that $z(t)=t$ (this variable is there only to ensure that the length of $z$ grows
at least linearly).
Let $t,\mu\in\Rp$ and assume that $\glen{z}(0,t)\geqslant\Omega^*(\infnorm{x},\mu)$.
We will show that $t\geqslant\Omega(\infnorm{x},\mu)$ by contradiction. Assume the contrary
and let $u\in[0,t]$. By definition:
\begin{align*}
\infnorm{y(u),z(u)}
    &\leqslant1+\infnorm{y(u)}\leqslant1+\Upsilon(\infnorm{x},t)\\
    &<1+\Upsilon(\infnorm{x},\Omega(\infnorm{x},\mu))
\end{align*}
and thus
\begin{align*}
\infnorm{y'(u),z'(u)}&=\infnorm{1,p(y(u))}\\
    &<1+\sigmap{p}\big(1+\Upsilon(\infnorm{x},\Omega(\infnorm{x},\mu)))^k.
\end{align*}
Consequently:
\[\glen{y,z}(0,t)<t\sup_{u\in[0,t]}\infnorm{y'(u),z'(u)}\leqslant\Omega^*(\infnorm{x},\mu)\]
which is absurd. Since $t\geqslant\Omega(\infnorm{x},\mu)$, by definition we get
that
\[\infnorm{y_{1..m}(t)-f(x)}\leqslant e^{-\mu}.\]
Finally, $\infnorm{y'(t),z'(t)}\geqslant\infnorm{z'(u)}\geqslant1$
for all $t\in\Rp$. This shows that that $f\in\glc{\Omega^*}$ where $\Omega^*$ is a polynomial.

Let $f\in\glc{\Omega}$ where $\Omega$ is a polynomial, which we assume to be an increasing function.

Apply \defref{def:glc} to get $\Omega, d,p,q$. Also assume that the
polynomial  $\Omega$ is an increasing function. Let $k=\degp{p}$.
Apply \lemref{lem:norm} to get that $g(x)=\norm_{\infty,1}(p(x))$ belongs to $\gpval$.
Apply \defref{def:gpac_generable_ext} to get the corresponding $m,r,x_0$ and $z_0$.
Let $x\in\dom{f}$. For the analysis, it will be useful to consider the following systems:
\[\left\{\begin{array}{@{}r@{}l}
y(0)&=q(x)\\z(x_0)&=z_0
\end{array}\right.\qquad
\left\{\begin{array}{@{}r@{}l}
y'(t)&=p(y(t))\\
\jacobian{z}(x)&=r(z(x))\end{array}\right.
\]
Note that by definition $z_1(x)=g(x)$. Define $\psi(t)=g(y(t))$
and $\hat{\psi}(u)=\int_0^u\psi(t)dt$. Now define the following system:
\[\left\{\begin{array}{@{}r@{}l}
\hat{y}(0)&=q(x)\\\hat{z}(0)&=z(q(x))\\\hat{w}(0)&=\frac{1}{g(q(x))}
\end{array}\right.\qquad
\left\{\begin{array}{@{}r@{}l}
\hat{y}'(u)&=\hat{w}(u)p(\hat{y}(u))\\
\hat{z}'(u)&=\hat{w}(u)r(\hat{z}(u))p(\hat{y}(u))\\
\hat{w}'(u)&=-\hat{w}(u)^3r_1(\hat{z}(u))p(\hat{y}(u))
\end{array}\right.
\]
where by $r_1$ we mean the first row of $r$. We will check that $\hat{y}(u)=y(\hat{\psi}^{-1}(u))$,
$\hat{z}(u)=z(\hat{y}(u))$ and $\hat{w}(u)=(\hat{\psi}^{-1})'(u)$. We will use the fact that
for any $h\in C^1$, $(h^{-1})'=\frac{1}{h'\circ h^{-1}}$. Also note that $\hat{\psi}'=\psi$.
\begin{itemize}
\item $\hat{y}(0)=y(\hat{\psi}^{-1}(0))=y(0)=q(x)$
\item $\hat{y}'(u)=(\hat{\psi}^{-1})'(u)y'(\hat{\psi}^{-1}(u))=\hat{w}(u)p(y(\hat{\psi}^{-1}(u)))=\hat{w}(u)p(\hat{y}(u))$
\item $\hat{z}(0)=z(\hat{y}(0))=z(q(x))$
\item $\hat{z}'(u)=\jacobian{z}(\hat{y}(u))\hat{y}'(u)=\hat{w}(u)r(z(\hat{y}(u)))p(\hat{y}(u))=\hat{w}(u)r(\hat{z}(u))p(\hat{y}(u))$
\item $\hat{w}(0)=\frac{1}{\hat{\psi}'(\hat{\psi}^{-1}(0))}=\frac{1}{\psi(0)}=\frac{1}{g(q(x))}$
\item $\hat{w}'(u)=\frac{-(\hat{\psi}^{-1})'(u)\hat{\psi}''(\hat{\psi}^{-1}(u))}{\cramped{(\hat{\psi}'(\hat{\psi}^{-1}(u)))^2}}
    =-\hat{w}(u)^3\psi'(\hat{\psi}^{-1}(u))=\scalarprod{\grad{g}(y(\hat{\psi}^{-1}(u)))}{y'(\hat{\psi}^{-1})}$
    and since $\grad{g}(x)=\transpose{r_1(z(x))}$ (transpose of the first row of the Jacobian
matrix of $z$ because $g=z_1$) then\[
\hat{w}'(u)
    =-\hat{w}(u)^3\scalarprod{\transpose{r_1(z(y(\hat{\psi}^{-1}(u))))}}{p(y(\hat{\psi}^{-1}(u)))}
    =-\hat{w}(u)^3r_1(\hat{z}(u))p(\hat{y}(u))
\]
\end{itemize}

We now claim that this system computes $f$ quickly and has a polynomial bound.
First note that by \lemref{lem:norm}:
\[\infnorm{y'(t)}\leqslant g(y(t))\leqslant\infnorm{y'(t)}+1\]
and thus
\[\glen{y}(0,t)\leqslant\hat{\psi}(t)\leqslant\glen{y}(0,t)+t.\]
Thus
\begin{align*} \glen{\hat{y}}(0,u)&=\int_0^u\infnorm{\hat{y}'(\xi)}d\xi=
\int_0^{\hat{\psi}^{-1}(u)}\infnorm{\hat{w}(\hat{\psi}(t))p(\hat{y}(\hat{\psi}(t)))}\hat{\psi}'(t)dt
\\&=\int_0^{\hat{\psi}^{-1}(u)}\infnorm{(\hat{\psi}^{-1})'(\hat{\psi}(t))\hat{\psi}'(t)p(y(t))}dt \\
&=\int_0^{\hat{\psi}^{-1}(u)}\infnorm{p(y(t))}dt=\glen{y}(0,\hat{\psi}^{-1}(u))\leqslant\hat{\psi}(\hat{\psi}^{-1}(u))\leqslant
u.
\end{align*}
It follows that
\[\infnorm{\hat{y}(u)}\leqslant\infnorm{\hat{y}(0)}+u\leqslant\infnorm{q(x)}+u\leqslant\poly(\infnorm{x},u).\]
Similarly:
\[\infnorm{\hat{z}(u)}=\infnorm{z(\hat{y}(u))}\leqslant\poly(\infnorm{x},u)\]
because $z\in\gpval$ and thus is polynomially bounded. Finally,
\[\infnorm{\hat{w}}=\frac{1}{\psi(\hat{\psi}^{-1}(u)}=\frac{1}{g(\hat{y}(u))}\leqslant\frac{1}{\infnorm{y'(\hat{\psi}^{-1}(u))}}\leqslant1\]
because by hypothesis, $\infnorm{y'(t)}\geqslant1$ for all $t\in\Rp$. This shows that
indeed $\infnorm{(\hat{y},\hat{z},\hat{w})(u)}$ is polynomially bounded in $\infnorm{x}$ and $u$.
Now let $\mu\in\Rp$ and $t\geqslant1+\Omega(\infnorm{x},\mu)$ then
\begin{align*}
\glen{\hat{y}}(0,t)&=\glen{y}(0,\hat{\psi}^{-1}(t))\\
&\geqslant\hat{\psi}(\hat{\psi}^{-1}(t))-\hat{\psi}^{-1}(t)\\
&\geqslant t-\hat{\psi}^{-1}(t)\\
&\geqslant 1+\Omega(\infnorm{x},\mu)-\frac{1}{\psi(\hat{\psi}^{-1}(t))}\\
&\geqslant \Omega(\infnorm{x},\mu)
\end{align*}
because, as we already saw, $\infnorm{\psi(\hat{\psi}^{-1}(t))}\geqslant1$.
Thus by definition:
\[\infnorm{\hat{y}_{1..m}(t)-f(x)}\leqslant e^{-\mu}\]
because $\hat{y}(t)=y(\hat{\psi}^{-1}(t))$.
This shows that $f\in\gpc$.

\section{Proof that ALP implies AWP} 

The purpose of the current section is to state the following.

\begin{theorem}
$\gpc=\gpwc$.
\end{theorem}

\begin{proof}
The inclusion $\gpc\subseteq\gpwc$ is immediate from Definitions \ref{def:gc} and \ref{def:gwc}. The other inclusion will follow from the results of the other sections.
\end{proof}

\section{Proof that AWP implies ARP}
\label{app:weak_to_robust}

The purpose of the current section is to state the following:

\begin{theorem}  \label{th:weak_to_robust}
$\gpwc\subseteq\gprc$.
\end{theorem}

i.e. that it possible to transform a computation into a computation
that tolerates small perturbations of the dynamics, where:

\begin{definition}[Analog robust computability]\label{def:grc}
Let $n,m\in\N$, $f:\subseteq\R^n\rightarrow\R^m$, $\Theta,\Omega:\Rp^2\rightarrow\Rp$
and $\Upsilon:\Rp^3\rightarrow\Rp$.
We say that $f$ is $(\Upsilon,\Omega,\Theta)$-robustly-computable if and only if there exist $d\in\N$,
and $(h:\R^d\rightarrow\R^d),(g:\R^n\times\Rp\rightarrow\R^d)\in\gpval$
such that for any $x\in\dom{f}$, $\mu\in\Rp$, $e_0\in\R^d$ and $e\in C^0(\Rp,\R^d)$
satisfying
\[\infnorm{e_0}+\int_0^\infty\infnorm{e(t)}dt\leqslant e^{-\Theta(\infnorm{x},\mu)},\]
there exists (a unique) $y:\Rp\rightarrow\R^d$ satisfying for all $t\in\Rp$:
\begin{itemize}
\item $y(0)=g(x,\mu)+e_0$ and $y'(t)=h(y(t))+e(t)$\hfill$\blacktriangleright$ $y$ satisfies a generable IVP
\item if $t\geqslant\Omega(\infnorm{x},\mu)$ then $\infnorm{y_{1..m}(t)-f(x)}\leqslant e^{-\mu}$
\hfill$\blacktriangleright$ $y_{1..m}$ approximates $f(x)$
\item $\infnorm{y(t)}\leqslant\Upsilon(\infnorm{x},\mu,t)$
\hfill$\blacktriangleright$ $y(t)$ is bounded
\end{itemize}
We denote by $\grc{\Upsilon}{\Omega}{\Theta}$ the set of
$(\Upsilon,\Omega,\Theta)$-robustly-computable functions, and by $\gprc$ the set of
$(\poly,\poly,\poly)$-robustly-computable functions.
\end{definition}

Intuitively, this definition says that even if the initial condition and the ODE defining the PIVP are
(slightly) perturbed or have (small) errors in \defref{def:gc}, the PIVP is still capable of
computing an approximation of $f(x)$.

Actually, we prove in this section that $\gpwc\subseteq\gprc$. Then the
equality will follow from results of other sections. 

\subsection{Some remarks}

\begin{remark}[Domain of definition of $g$ and $h$]
There is a subtle but important detail in this definition: we more or less replaced
the polynomials $p$ and $q$ by generable functions $g$ and $h$. It could have been
tempting to take this opportunity to restrict the domain of definition of $g$ to $\dom{f}\times\Rp$
and that of $h$ to a subset of $\R^d$ where the dynamics takes place. We kept the entire
euclidean space for good reasons. First it makes the definition simpler. Second,
it makes the notion stronger and more useful. This last point is important because
we are going to use robust computability (and the next notion of strong computability)
in cases where we have less or no control over the errors and thus over the trajectory of the
system. On the downside, this requires to check that $g$ and $h$ are indeed defined
over the entire space !
\end{remark}

The examples below show how to build robustly-computable functions. In the first example,
we only need to define $\Theta$ so that it works, whereas in the second case, careful
design of the system is needed for it to be robust.

\begin{example}[Polynomials are robustly-computable]
In order to make polynomials robustly-computable, we will play with the choice of $\Theta$
and see that this is enough to make the system robust.
Let $q\in\K[\R^d]$ be a multivariate polynomial: we will show that $q\in\gprc$.
Let $x\in\R^d$, $\mu\in\Rp$, $e_0\in\R$ and $e\in C^0(\Rp,\R)$.
Assume that $|e_0|+\int_0^\infty|e(t)|dt\leqslant e^{-\mu}$ and consider the following system for $t\in\Rp$:
\[y(0)=q(x)+e_0\qquad y'(t)=e(t)\]
We claim that this system satisfies \defref{def:grc}:
\begin{itemize}
\item The system is of the form $y(0)=\poly(x)+e_0$ and $y'(t)=\poly(y(t))+e(t)$ where the polynomials
have coefficients in $\K$.
\item For any $t\geqslant0$, we have:
\[\infnorm{y(t)-q(x)}\leqslant |e_0|+\int_0^t|e(u)|du\leqslant|e_0|+\int_0^\infty|e(u)|du\leqslant e^{-\mu}\]
so we can take $\Omega(\alpha,\mu)=0$.
\item For any $t\in\Rp$, we have:
\[\infnorm{y(t)}\leqslant\infnorm{q(x)}+|e_0|+\int_0^t|e(u)|du\leqslant\infnorm{q(x)}+1\leqslant\poly(\infnorm{x})\]
so we can take $\Upsilon$ to be any polynomial such that $\Upsilon(\infnorm{x},\mu)\geqslant\infnorm{p(x)}+1$.
\end{itemize}
This shows that $q\in\grc{\Upsilon}{\Omega}{\Theta}$ where $\Theta(\alpha,\mu)=\mu$.
\end{example}

In the previous example, we saw that we could modify the associated system of some computable
functions to make them robustly-computable. It appears that this is not a coincidence
but a general fact. To understand how the proof works, one must first understand the problem.
Let us consider a computable function $f:\subseteq\R^d\rightarrow\R$ in $\gwc{\Upsilon}{\Theta}$ and the associated system
for $x\in\dom{f}$ and $\mu\in\Rp$:
\[y(0)=q(x,\mu)\qquad y'(t)=p(y(t))\]
This system converges to $f(x)$ very quickly: $\infnorm{y_1(t)-f(x)}\leqslant e^{-\mu}$
when $t\geqslant\Omega(\infnorm{x},\mu)$ and $y(t)$ is bounded: $\infnorm{y(t)}\leqslant\Upsilon(\infnorm{x},\mu,t)$.
Let us introduce some errors in the system
by taking $e_0\in\R^d$ and $e\in C^0(\Rp,\R^d)$ such that $\infnorm{e_0}+\int_0^\infty\infnorm{e(u)}du\leqslant e^{-\Theta(\infnorm{x},\mu)}$
for some unspecified $\Theta$ and consider the perturbed system:
\[z(0)=q(x,\mu)+e_0\qquad z'(t)=p(z(t))+e(t)\]
The relationship between this system and the previous one is given by
\thref{th:dependency_right_side}
and can be informally written as:
\begin{align}\label{eq:rel_weak:informal_error}
\infnorm{z(t)-y(t)}&\leqslant\left(\infnorm{e_0}+\int_0^t\infnorm{e(u)}du\right)e^{\int_0^tk\sigmap{p}\infnorm{y(u)}^{k-1}du}\\
    &\leqslant\left(\infnorm{e_0}+\int_0^\infty\infnorm{e(u)}du\right)e^{\int_0^tk\sigmap{p}\Upsilon(\infnorm{x},\mu,u)^{k-1}du}\tag*{using the bound of $y(t)$}\\
    &\leqslant e^{k\sigmap{p}t\Upsilon(\infnorm{x},\mu,t)^{k-1}-\Theta(\infnorm{x},\mu)}\tag*{assuming that $\Upsilon$ is increasing}
\end{align}
One observes that this bound grows to infinity whatever we choose for $\Theta$ because of the
dependency in $t$. On the other hand, we do not need to simulate $y$ for arbitrary
large $t$: as soon as $t\geqslant\Theta(\infnorm{x},\mu)$ we can stop the system and get
a good enough result. Unfortunately, one does not simply stop a differential system, however
we can slow it down 
. To this end, introduce
$\psi(t)=(1+\Theta(\infnorm{x},\mu))\tanh(t)$ and $w(t)=z(\psi(t))$. If we show that $w$ 
satisfies a differential system, then we are almost done. Indeed $\psi(t)\leqslant1+\Theta(\infnorm{x},\mu)$ for all $t\in\Rp$
and if $t\geqslant1+\Theta(\infnorm{x},\mu)$ then $\psi(t)\geqslant\Theta(\infnorm{x},\mu)$,
so the system ``kind of stops'' between $\Theta(\infnorm{x},\mu)$ and $\Theta(\infnorm{x},\mu)+1$.
Furthermore, if  $t\geqslant1+\Theta(\infnorm{x},\mu)$ then:
\begin{align*}
\infnorm{w_1(t)-f(x)}&\leqslant\infnorm{z(\psi(t))-y(\psi(t))}+\infnorm{y_1(\psi(t))-f(x)}\tag*{use the triangle inequality}\\
    &\leqslant e^{k\sigmap{p}\psi(t)\Upsilon(\infnorm{x},\mu,\psi(t))^{k-1}-\Theta(\infnorm{x},\mu)}+e^{-\mu}
    \tag*{using \eqref{eq:rel_weak:informal_error}}\\
    &\leqslant e^{k\sigmap{p}(1+\Theta(\infnorm{x},\mu))\Upsilon(\infnorm{x},\mu,1+\Theta(\infnorm{x},\mu))^{k-1}-\Theta(\infnorm{x},\mu)}+e^{-\mu}\tag*{using the bound on $\psi$}\\
    &\leqslant 2e^{-\mu}\tag*{for a suitable choice of $\Theta$}
\end{align*}
We are left with showing that $w(t)=z(\psi(t))$ can be be generated by a generable IVP
with perturbations. In the case of no perturbations, this is very easy because
$w'(t)=\psi'(t)z'(t)=x(1-\tanh(t))p(z(t))$ which is generable. The following lemma
extends this idea to the case of perturbations.

\begin{lemma}[PIVP Slow-Stop]\label{lem:pivp_slowstop}
Let $d\in\N$, $y_0\in\R^d$, $T,\theta\in\Rp$, $(e_{0,y},e_{0,A})\in\R^{d+1}$,
$(e_y,e_A)\in C^0(\Rp,\R^{d+1})$ and $p\in\K^d[\R^d]$.
Assume that $\infnorm{e_0}+\int_0^\infty\infnorm{e(t)}dt\leqslant e^{-\theta}$
and consider the following system:
\[
\left\{\begin{array}{@{}r@{}l}y(0)&=y_0+e_{0,y}\\
A(0)&=T+2+e_{0,A}\end{array}\right.
\quad
\left\{\begin{array}{@{}r@{}l}y'(t)&=\frac{1+\tanh(A(t))}{2}p(y(t))+e_{y}(t)\\
A'(t)&=-1+e_{A}(t)\end{array}\right.
\]
Then there exist an increasing function $\psi\in C^0(\Rp,\Rp)$
and $z:\psi(\Rp)\rightarrow\R^d$ such that:
\[\psi(0)=0\qquad z(0)=y_0+e_{0,y}\qquad z'(t)=p(z(t))+(\psi^{-1})'(t)e_{y}(\psi^{-1}(t))\]
and $y(t)=z(\psi(t))$. Furthermore $\psi(T+1)\geqslant T$
and $\psi(t)\leqslant T+4$ for all $t\in\Rp$. Furthermore, $|A(t)|\leqslant T+3$ for
all $t\in\Rp$.
\end{lemma}

\begin{proof}
Let $f(t)=\frac{1+\tanh(A(t))}{2}$ and note that $0<f(t)<1$ for all $t\in\Rp$.
Check that we can integrate $A$ explicitly:
\[A(t)=T+2-t+e_{0,A}+\int_0^te_A(u)du.\]
If we take $\psi(t)=\int_0^tf(u)du$ then $\psi$ is an increasing function because $f>0$,
so it is a diffeomorphism from $\Rp$ onto $\psi(\Rp)$.
Note that $\psi(t)\leqslant t$ for all $t\in\Rp$. Let $t\geqslant T+3$,
then
\begin{align*}
A(t)&\leqslant T+2-t+|e_{0,A}|+\int_0^t|e_A(u)|du\\
&\leqslant T+2+e^{-\theta}-t\leqslant T+3-t\leqslant0
\end{align*}
because $\theta\geqslant0$. Apply \lemref{lem:bounds_tanh} to get that
$\tanh(A(t))\leqslant -1+e^{T+3-t}$ and thus
\[f(t)\leqslant \frac{1}{2}e^{T+3-t}\text{ for }t\geqslant T+3.\]
Integrating this inequality shows that
\begin{align*}
\psi(t)&\leqslant\psi(T+3)+\frac{1}{2}\int_{T+3}^{t}e^{T+3-u}du\\
&\leqslant T+3+\frac{1}{2}(1-e^{T+3-t})\leqslant T+4.
\end{align*}
This shows that $\psi(t)\leqslant T+4$
for all $t\in\Rp$.

Let $t\leqslant T+1$, then by the same reasoning:
\[A(t)\geqslant T+2-t-e^{-\theta}\geqslant T+1-t\geqslant0\]
thus $\tanh(A(t))\geqslant 1-e^{t-T-1}$ and $f(t)\geqslant1-\frac{1}{2}e^{t-T-1}$.
Thus:
\[\psi(T+1)\geqslant\int_0^{T+1}1+\frac{1}{2}e^{u-T-1}du=T+1+\frac{1}{2}(1-e^{-1-T})\geqslant T.\]
Finally, apply \lemref{lem:perturbed_time_scale_gpac} to get that $y(t)=z(\psi(t))$
where $z$ satisfies for $t\in\psi(\Rp)$:
\[z(0)=y(0)\qquad z'(t)=p(z(t))+(\psi^{-1})'(t)e_{y}(\psi^{-1}(t))\]
\end{proof}



\subsection{The proof}

The proof of the implication AWP implies ARP of \thref{th:weak_to_robust} is then the following.

Let $\Upsilon^*,\Omega^*$ be polynomials such that $f\in\gwc{\Upsilon^*}{\Omega^*}$.
Without loss of generality, we assume they are increasing functions on both arguments.
Apply \defref{def:gwc} to get $d\in\N$, $p\in\K^d[\R^d]$, $q\in\K^d[\R^{n+1}]$ and let $k=\degp{p}$.
Define:
\begin{align*}
T(\alpha,\mu)&=\Omega^*(\alpha,\mu+\ln2)\\
\Theta(\alpha,\mu)&=k\sigmap{p}(T(\alpha+1,\mu)+4)(\Upsilon^*(\alpha,\mu,T(\alpha+1,\mu)+4)+1)^{k-1}+\mu+\ln2\\
\Omega(\alpha,\mu)&=T(\alpha+1,\mu)+1
\end{align*}
Let $x\in\dom{f}$, $(e_{0,y},e_{0,A})\in\R^{d+1}$, $(e_y,e_A)\in C^0(\Rp,\R^{d+1})$ and $\mu\in\Rp$ such that
\[\infnorm{e_0}+\int_0^\infty\infnorm{e(t)}dt\leqslant e^{-\Theta(\infnorm{x},\mu)}.\]
Apply \lemref{lem:pivp_slowstop} and consider the following systems (where $\psi$ is given
by the lemma):
\[
\left\{\begin{array}{@{}r@{}l}y(0)&=q(x,\mu)+e_{0,y}\\
A(0)&=T(\norm_{\infty,1}(x),\mu)+2+e_{0,A}\end{array}\right.
\quad
\left\{\begin{array}{@{}r@{}l}y'(t)&=\frac{1+\tanh(A(t))}{2}p(y(t))+e_{y}(t)\\
A'(t)&=-1+e_{A}(t)\end{array}\right.
\]
\[
\left\{\begin{array}{@{}r@{}l}z(0)&=q(x,\mu)+e_{0,y}\\
z'(t)&=p(z(t))+(\psi^{-1})'(t)e_{y}(\psi^{-1}(t))\end{array}\right.
\quad
\left\{\begin{array}{@{}r@{}l}w(0)&=q(x,\mu)\\
w'(t)&=p(w(t))\end{array}\right.
\]
By definition of $p$ and $q$, if $t\geqslant\Omega^*(\infnorm{x},\mu)$ then
$\infnorm{w_{1..m}(t)-f(x)}\leqslant e^{-\mu}$. Furthermore,
$\infnorm{w(t)}\leqslant\Upsilon^*(\infnorm{x},\mu,t)$ for all $t\in\Rp$.
Define $T^*=T(\norm_{\infty,1}(x),\mu)$.
Apply \lemref{lem:norm} to get
that
\[\infnorm{x}\leqslant\norm_{\infty,1}(x)\leqslant\infnorm{x}+1\]
and thus
\[T(\infnorm{x},\mu)\leqslant T^*\leqslant T(\infnorm{x}+1,\mu).\]
By construction, $\psi(t)\leqslant T^*+4$ for all $t\in\Rp$. Let $t\in\Rp$,
apply \thref{th:dependency_right_side} by checking that:
\begin{align*}
&\left(\infnorm{e_{0,y}}+\int_0^{\psi(t)}\infnorm{(\psi^{-1})'(u)e_{y}(\psi^{-1}(u))}du\right)
e^{k\sigmap{p}\int_0^{\psi(t)}(\infnorm{w(u)}+1)^{k-1}du}\\
&\leqslant\left(\infnorm{e_{0,y}}+\int_0^t\infnorm{e_{y}(u)}du\right)
e^{k\sigmap{p}\int_0^{\psi(t)}(\Upsilon^*(\infnorm{x},\mu,u)+1)^{k-1}du}\tag*{by a change of variable}\\
&\leqslant e^{k\sigmap{p}\psi(t)(\Upsilon^*(\infnorm{x},\mu,\psi(t))+1)^{k-1}-\Theta(\infnorm{x},\mu)}\tag*{by hypothesis on the error}\\
&\leqslant e^{k\sigmap{p}(T(\infnorm{x}+1,\mu)+4)(\Upsilon^*(\infnorm{x},\mu,T(\infnorm{x}+1,\mu)+4)+1)^{k-1}-\Theta(\infnorm{x},\mu)}\tag*{because $\psi$ is bounded}\\
&\leqslant e^{-\mu-\ln2}\leqslant1\tag*{by definition of $\Theta$}
\end{align*}
Thus
\[\infnorm{z(\psi(t))-w(\psi(t))}\leqslant e^{-\mu-\ln2}\text{ for all }t\in\Rp.\]
Furthermore, if $t\geqslant\Omega(\infnorm{x},\mu)$ then
\[\psi(t)\geqslant
\psi(T(\infnorm{x}+1,\mu)+1)\geqslant\psi(T^*+1)\geqslant T^*.\]
By construction $\psi(T^*)\geqslant T^*$
so
\[\psi(t)\geqslant T^*\geqslant T(\infnorm{x},\mu)=\Omega^*(\infnorm{x},\mu+\ln2)\]
thus
\[\infnorm{z(\psi(t))-f(x)}\leqslant e^{-\mu-\ln2}.\]
Consequently, we have 
\begin{align*}
\infnorm{y(t)-f(x)}&\leqslant\infnorm{z(\psi(t))-w(\psi(t))}+\infnorm{w(\psi(t))-f(x)}\\
&\leqslant 2e^{-\mu-\ln2}\leqslant e^{-\mu}.
\end{align*}

Let $t\in\Rp$, then
\begin{align*}
\infnorm{y(t)}&=\infnorm{z(\psi(t))}\leqslant\infnorm{w(\psi(t))}+e^{-\mu}\\
&\leqslant\Upsilon^*(\infnorm{x},\mu,\psi(t))+1\leqslant\Upsilon^*(\infnorm{x},\mu,T(\infnorm{x}+1,\mu)+4)+1\\
&\leqslant\Upsilon^*(\infnorm{x},\mu,\Omega^*(\infnorm{x}+1,\mu+\ln2)+4)+1
\end{align*}
which is polynomially bounded in $\infnorm{x}$ and $\mu$.
Furthermore
\[|A(t)|\leqslant T^*+4\leqslant\Omega^*(\infnorm{x}+1,\mu+\ln2)+4\]
which are both polynomially bounded in $\infnorm{x}$, $\mu$.

Finally, $(y,A)(0)=g(x,\mu)+e_0$ and $(y,A)'(t)=h(y(t),A(t))+e(t)$
where $g$ and $h$ belong to $\gpval[\K]{}$ because $\tanh,\norm_{\infty,1}\in\gpval$.

\begin{remark}[Polynomial versus generable]\label{rem:weak_to_robust_gen}
The proof of \thref{th:weak_to_robust} also works if $q$ is generable (i.e. $q\in\gpval$)
instead of polynomial in \defref{def:gc} or \defref{def:gwc}.
\end{remark}

\section{Proof that ARP implies ASP}
\label{app:robust_to_strong}

This section is devoted to prove the following result:  it is always possible to avoid that the system in \defref{def:grc}.
explodes for inputs not in the domain of the function, or for perturbations of the dynamics
which are too big. This motivates the following result and \defref{def:gsc}.

\begin{theorem}[Robust $\subseteq$ strong]\label{th:robust_to_strong}
$\gprc=\gpsc$.
\end{theorem}

where

\begin{definition}[Analog strong computability]\label{def:gsc}
Let $n,m\in\N$, $f:\subseteq\R^n\rightarrow\R^m$, $\Theta,\Omega:\Rp^2\rightarrow\Rp$
and $\Upsilon:\Rp^4\rightarrow\Rp$.
We say that $f$ is $(\Upsilon,\Omega,\Theta)$-strongly-computable if and only if there exist $d\in\N$,
and $(h:\R^d\rightarrow\R^d),(g:\R^n\times\Rp\rightarrow\R^d)\in\gpval$
such that for any $x\in\R^n$, $\mu\in\Rp$, $e_0\in\R^d$ and $e\in C^0(\Rp,\R^d)$,
there is (a unique) $y:\Rp\rightarrow\R^d$ satisfying for all $t\in\Rp$
and $\hat{e}(t)=\infnorm{e_0}+\int_0^t\infnorm{e(u)}du$:
\begin{itemize}
\item $y(0)=g(x,\mu)+e_0$ and $y'(t)=h(y(t))+e(t)$\hfill$\blacktriangleright$ $y$ satisfies a generable IVP
\item if $x\in\dom{f}$,  $t\geqslant\Omega(\infnorm{x},\mu)$ and $\hat{e}(t)\leqslant e^{-\Theta(\infnorm{x},\mu)}$
then $\infnorm{y_{1..m}(t)-f(x)}\leqslant e^{-\mu}$
\item $\infnorm{y(t)}\leqslant\Upsilon(\infnorm{x},\mu,\hat{e}(t),t)$\hfill$\blacktriangleright$ $y(t)$ is bounded
\end{itemize}
We denote by $\gsc{\Upsilon}{\Omega}{\Theta}$ the set of
$(\Upsilon,\Omega,\Theta)$-strongly-computable functions, and by $\gpsc$ the set of
$(\poly,\poly,\poly)$-strongly-computable functions.
\end{definition}

Actually, we prove in this section that
$\gprc\subseteq\gpsc$. Equality follows from results in other
sections. 

\subsection{Some remarks}

The following Lemma can be proved by providing explicitly such a
function: 

\begin{lemma}[Max function, \citep{BGP16Arxiv}]\label{lem:max}
There is a family of functions $\mx_\delta\in\gpval$ such that: 
For any $x,y\in\R$ and $\delta\in]0,1]$ we have:
\[\max(x,y)\leqslant\mx_\delta(x,y)\leqslant\max(x,y)+\delta\]
For any $x\in\R^n$ and $\delta\in]0,1]$ we have:
\[\max(x_1,\ldots,x_n)\leqslant\mx_\delta(x)\leqslant\max(x_1,\ldots,x_n)+\delta\]
\end{lemma}

The following lemmas can also be established:

\begin{lemma}[Bounds on $\tanh$, \citep{BGP16Arxiv}]\label{lem:bounds_tanh}
$1-\sgn{t}\tanh(t)\leqslant e^{-|t|}$ for all $t\in\R$.
\end{lemma}

\begin{lemma}[Perturbed time-scaling]\label{lem:perturbed_time_scale_gpac}
Let $d\in\N$, $x_0\in\R^d$, $p\in\R^d[\R^d]$, $e\in C^0(\Rp,\R^d)$ and $\phi\in C^0(\Rp,\Rp)$.
Let $\psi(t)=\int_0^t\phi(u)du$.
Assume that $\psi$ is an increasing function and that $y,z:\Rp\rightarrow\R^d$ satisfy for all $t\in\Rp$:
\[\left\{\begin{array}{@{}r@{}l}y(0)&=x_0\\y'(t)&=p(y(t))+(\psi^{-1})'(t)e(\psi^{-1}(t))\end{array}\right.\qquad
\left\{\begin{array}{@{}r@{}l}z(0)&=x_0\\z'(t)&=\phi(t)p(z(t))+e(t)\end{array}\right.\]
Then $z(t)=y\left(\psi(t)\right)$ for all $t\in\Rp$. In particular,
\[\int_0^{\psi(t)}\infnorm{(\psi^{-1})'(u)e(\psi^{-1}(u))}du=\int_0^t\infnorm{e(u)}du\]
and
\[\sup_{u\in[0,\psi(t)]}\infnorm{(\psi^{-1})'(u)e(\psi^{-1}(u))}=\sup_{u\in[0,t]}\frac{\infnorm{e(u)}}{\phi(u)}.\]
\end{lemma}

\begin{proof}
Use that $\phi=\psi'$, $\psi'\cdot(\psi^{-1})'\circ\psi=1$ and that $\psi'\geqslant0$.
\end{proof}

On a more technical side, we will need to ``apply'' \defref{def:grc} over finite intervals
and we need the following lemma to do so.

\begin{lemma}[Finite time robustness]\label{lem:finite_time_robust}
Let $f\in\grc{\Upsilon}{\Omega}{\Theta}$, $I=[0,T]$, $x\in\dom{f}$, $\mu\in\Rp$, $e_0\in\R^d$
and $e\in C^0(I,\R^d)$ such that
\[\infnorm{e_0}+\int_I\infnorm{e(t)}dt<e^{-\Theta(\infnorm{x},\mu)}.\]
Assume that $y:I\rightarrow\R^d$ satisfies for all $t\in I$:
\[y(0)=g(x,\mu)+e_0\qquad y'(t)=h(y(t))+e(t)\]
where $g,h$ come from \defref{def:grc} applied to $f$. Then for all $t\in I$:
\begin{itemize}
\item $\infnorm{y(t)}\leqslant\Upsilon(\infnorm{x},\mu,t)$
\item if $t\geqslant\Omega(\infnorm{x},\mu)$ then $\infnorm{y_{1..m}-f(x)}\leqslant e^{-\mu}$
\end{itemize}
\end{lemma}

\begin{proof}
The trick is simply to extend $e$ so that it is defined over $\Rp$ and such that:
\[\infnorm{e_0}+\int_0^\infty\infnorm{e(u)}du\leqslant e^{-\Theta(\infnorm{x},\mu)}\]
This is always possible because the truncated integral is strictly smaller than the bound.
Formally, define for $t\in\Rp$:
\[\bar{e}(t)=\begin{cases}e(t)&\text{if }t\leqslant T\\
    e(T)e^{\frac{e(T)}{\varepsilon}(T-t)}&\text{otherwise}\end{cases}
\]
\[
\qquad\text{where }\varepsilon=e^{-\Theta(\infnorm{x},\mu)}-\infnorm{e_0}-\int_I\infnorm{e(t)}dt>0\]
One easily checks that $\bar{e}\in C^0(\Rp,\R^d)$ and that:
\begin{align*}
\infnorm{e_0}+\int_0^\infty\infnorm{\bar{e}(t)}dt&=\infnorm{e_0}+\int_0^T\infnorm{e(t)}dt+\int_T^\infty e(T)e^{\frac{e(T)}{\varepsilon}(T-t)}dt\\
    &=e^{-\Theta(\infnorm{x},\mu)}-\varepsilon+\left[-\varepsilon e(T)e^{\frac{e(T)}{\varepsilon}(T-t)}\right]_T^\infty\\
    &=e^{-\Theta(\infnorm{x},\mu)}
\end{align*}
Assume that $z:\Rp\rightarrow\R^d$ satisfies for $t\in\Rp$:
\[z(0)=g(x,\mu)\qquad z'(t)=g(z(t))+\bar{e}(t)\]
Then $z$ satisfies \defref{def:grc} so $\infnorm{z}(t)\leqslant\Upsilon(\infnorm{x},\mu)$
and if $t\geqslant\Omega(\infnorm{x},\mu)$ then $\infnorm{z_{1..m}(t)-f(x)}\leqslant e^{-\mu}$.
Conclude by noting that $z(t)=y(t)$ for all $t\in[0,T]$ since $e(t)=\bar{e}(t)$.
\end{proof}

\subsection{The proof}

The proof of \thref{th:robust_to_strong} is then the following.

\begin{proof}
Let $\Omega,\Theta,\Upsilon$ be polynomials and $(f:\subseteq\R^n\rightarrow\R^m)\in\grc{\Upsilon}{\Omega}{\Theta}$.
Without loss of generality, we assume that $\Omega$, $\Theta$, $\Upsilon$ are
increasing functions of their arguments.
Apply \defref{def:grc} to get $d$, $h$ and $g$.
Let $x\in\R^n$, $\mu\in\Rp$, $(e_{0,y},e_{0,\ell})\in\R^{d+1}$ and $(e_y,e_\ell)\in C^0(\Rp,\R^{d+1})$.
Define $\hat{e}(t)=\infnorm{e_0}+\int_0^t\infnorm{e(u)}du$,
and consider the following system for $t\in\Rp$:
\[
\left\{\begin{array}{@{}r@{}l}
y(0)&=g(x,\mu)+e_{0,y}\\
y'(t)&=\psi(t)h(y(t))+e_{y}(t)\\
\ell(0)&=\mx_1(\norm_{\infty,1}(x),\mu)+1+e_{0,\ell}\\
\ell'(t)&=1+e_{\ell}(t)
\end{array}\right.
\]
\[\psi(t)=\frac{1+\tanh(\Delta(t))}{2}\qquad
\Delta(t)=\Upsilon(\ell(t),\ell(t),\ell(t))+1-\norm_{\infty,1}(y(t))\]
We will first show that the system remains polynomially bounded. Apply
\lemref{lem:max} and \lemref{lem:norm} to get that:
\begin{align*}
\infnorm{\ell(0)}&\leqslant\max(\infnorm{x}+1,\mu)+1+\infnorm{e_{0,\ell}}\\
    &\leqslant\poly(\infnorm{x},\mu)+\infnorm{e_{0,\ell}}
\end{align*}
Consequently:
\begin{align}
\infnorm{\ell(t)}&\leqslant\infnorm{\ell(0)}+\int_0^t1+\infnorm{e_{\ell}(u)}du\nonumber\\
    &\leqslant \poly(\infnorm{x},\mu)+t+\infnorm{e_{0,\ell}}+\int_0^t\infnorm{e_{\ell}(u)}du\nonumber\\
    &\leqslant\poly(\infnorm{x},\mu)+t+\hat{e}(t)\nonumber\\
    &\leqslant\poly(\infnorm{x},\mu,t,\hat{e}(t))\label{eq:robust_to_strong:ell}
\end{align}
Since $g,h\in\gpval$,
there exist two polynomials $\mtt{sp}$ and $\ovl{\mtt{sp}}$ such that $\infnorm{g(x)}\leqslant\mtt{sp}(\infnorm{x})$
and $\infnorm{h(x)}\leqslant\ovl{\mtt{sp}}(\infnorm{x})$
for all $x\in\R^d$ and without loss of generality, we assume that $\mtt{sp}$ and $\ovl{\mtt{sp}}$ are
increasing functions. Let $t\in\Rp$, there are two possibilities:
\begin{itemize}
\item If $\Delta(t)\geqslant0$ then $\norm_{\infty,1}(y(t))\leqslant1+\Upsilon(\ell(t),\ell(t),\ell(t))$
so apply  \lemref{lem:norm} and use \eqref{eq:robust_to_strong:ell} to conclude that $\infnorm{y(t)}\leqslant\poly(\infnorm{x},\mu,t,\hat{e}(t))$
and thus:
\begin{align}\infnorm{\psi(t)h(y(t))}&\leqslant\ovl{\mtt{sp}}(\infnorm{y(t)})\tag*{use that $\tanh<1$}\\
    &\leqslant\poly(\infnorm{x},\mu,t,\hat{e}(t))\label{eq:robust_to_strong:rhs_one}
\end{align}
\item If $\Delta(t)<0$ then apply \lemref{lem:bounds_tanh} to get that $\psi(t)\leqslant\frac{1}{2}e^{\Delta(t)}\leqslant e^{\Delta(t)}$.
Apply \lemref{lem:norm} to get that $\Delta(t)\leqslant\Upsilon(\ell(t),\ell(t),\ell(t))+1-\infnorm{y(t)}$
and thus $\infnorm{y(t)}\leqslant \Upsilon(\ell(t),\ell(t),\ell(t))+1-\Delta(t)$ and thus:
\begin{align}
\infnorm{\psi(t)h(y(t))}&\leqslant e^{\Delta(t)}\ovl{\mtt{sp}}(\infnorm{y(t)})\tag*{use the bound on $\psi$}\\
    &\leqslant e^{\Delta(t)}\ovl{\mtt{sp}}(\Upsilon(\ell(t),\ell(t),\ell(t))+1-\Delta(t))\tag*{use the bound on $\infnorm{y(t)}$}\\
    &\leqslant\poly(\ell(t))e^{\Delta(t)}\poly(-\Delta(t))\tag*{use that $\Upsilon$ is polynomial}\\
    &\leqslant\poly(\ell(t))\tag*{use that $e^{-x}\poly(x)=\bigO{1}$ for $x\geqslant0$ and fixed $\poly$}\\
    &\leqslant\poly(\infnorm{x},\mu,t,\hat{e}(t))\label{eq:robust_to_strong:rhs_two}
\end{align}
\end{itemize}
Putting \eqref{eq:robust_to_strong:rhs_one} and \eqref{eq:robust_to_strong:rhs_two} together, we get that:
\begin{align*}
\infnorm{y(t)}&\leqslant\infnorm{g(x,\mu)}+\infnorm{e_{0,y}}+\int_0^t\infnorm{\psi(u)h(y(u))}+\infnorm{e_y(u)}du\\
    &\leqslant\mtt{sp}(\infnorm{x,\mu})+\int_0^t\poly(\infnorm{x},\mu,u,\hat{e}(u))du+\hat{e}(t)\\
    &\leqslant\poly(\infnorm{x},\mu,t,\hat{e}(t))
\end{align*}
We will now analyze the behavior of the system when the error is bounded.
Define $\Theta^*(\alpha,\mu)=\Theta(\alpha,\mu)+1$. Define $\hat{\psi}(t)=\int_0^t\psi(u)du$ and note that it is a diffeomorphism
since $\psi>0$. Apply \lemref{lem:perturbed_time_scale_gpac} to get that $y(t)=z(\hat{\psi}(t))$
for all $t\in\Rp$, where $z$ satisfies for $\xi\in\hat{\psi}(\Rp)$:
\[z(0)=g(x,\mu)+e_{0,y}\qquad z'(\xi)=h(z(\xi))+\tilde{e}(\xi)\]
\[\qquad\text{where}
\int_0^{\hat{\psi}(t)}\infnorm{\tilde{e}(\xi)}d\xi=\int_0^t\infnorm{e_y(u)}du\]
Assume that $x\in\dom{f}$ and let $T\in\Rp$ such that $\hat{e}(T)\leqslant e^{-\Theta^*(\infnorm{x},\mu)}$.
Then $\hat{e}(T)<e^{-\Theta(\infnorm{x},\mu)}$ and for all $t\in[0,T]$:
\begin{align*}
\infnorm{e_{0,y}}+\int_0^{\hat{\psi}(t)}\infnorm{\tilde{e}}(u)du&=\infnorm{e_{0,y}}+\int_0^t\infnorm{e_y(u)}du\\
    &\leqslant\hat{e}(t)\leqslant e^{-\Theta(\infnorm{x},\mu)}
\end{align*}
Apply \lemref{lem:finite_time_robust} to get for all $u\in[0,\hat{\psi}(T)]$:
\begin{equation}\infnorm{z(u)}\leqslant\Upsilon(\infnorm{x},\mu,u)\end{equation}
\begin{equation}\text{if }u\geqslant\Omega(\infnorm{x},\mu)\text{ then }\infnorm{z_{1..m}(u)-f(x)}\leqslant e^{-\mu}\end{equation}
Apply Lemmas \ref{lem:max} and \ref{lem:norm} to get for all $t\in[0,T]$:
\begin{align*}
\ell(t)&\geqslant \mx_1(\norm_{\infty,1}(\infnorm{x},\mu))+1-\infnorm{e_{0,\ell}}+t-\int_0^t\infnorm{e_\ell(u)}du\\
    &\geqslant \max(\infnorm{x},\mu)+1+t-\hat{e}(t)\\
    &\geqslant\max(\infnorm{x},\mu,t)\tag*{using that $\hat{e}(t)\leqslant1$}
\end{align*}
Consequently, using \lemref{lem:norm}, for all $t\in[0,T]$:
\begin{align*}\label{eq:robust_to_strong:ineq_Delta}
\Delta(t)&\geqslant\Upsilon(\ell(t),\ell(t),\ell(t))-\infnorm{y(t)}\\
    &\geqslant\Upsilon(\infnorm{x},\mu,t)-\infnorm{y(t)}\tag*{using that $\ell(t)\geqslant\max(\infnorm{x},\mu,t)$}\\
    &=\Upsilon(\infnorm{x},\mu,t)-\infnorm{z(\hat{\psi}(t))}\tag*{using that $y(t)=z(\hat{\psi}(t))$}\\
    &\geqslant0\tag*{because $\hat{\psi}(t)\in[0,\hat{\psi}(T)]$}
\end{align*}
Consequently for all $t\in[0,T]$:
\[\hat{\psi}(t)=\int_0^t\psi(u)du=\int_0^t\frac{1+\tanh(\Delta(u))}{2}du\geqslant\frac{t}{2}\]
Define $\Omega^*(\alpha,\mu)=2\Omega(\alpha,\mu)$.
Assume that $T\geqslant \Omega^*(\infnorm{x},\mu)$ then $\hat{\psi}(T)\geqslant\Omega(\infnorm{x},\mu)$
and thus $\infnorm{y_{1..m}(T)-f(x)}=\infnorm{z(\hat{\psi}(T))-f(x)}\leqslant e^{-\mu}$.

Finally, $(y,\ell)(0)=g^*(x,\mu)+e_0$ where $g^*\in\gpval$. Similarly $(y,\ell)'(t)=h^*((y,\ell)(t))+e(t)$
where $h^*\in\gpval$. Note again that both $h^*$ and $g^*$ are defined over the entire space.
This concludes the proof that $f\in\gsc{\Omega^*}{\poly}{\Theta^*}$.
\end{proof}

\section{Proof that ASP implies AXP}
\label{app:strong_to_unaware}

This section is devoted to prove the following: in \defref{def:gsc} we defined a class
with a high degree of robustness to perturbations and related it to previous classes.
However, the value $f(x)$ the system computes
still depends on the initial condition (i.e.~$x$ is provided via the initial condition).
Here we want robustness to errors like in \defref{def:gsc}, but we also 
want to dynamically change the argument $x$ during a computation, as done in
\defref{def:goc}. Since these are two exigent requirements, we named this
computability form as ``extreme''. Here $\indicator{X}$ denotes the function defined by
$\indicator{X}(x)=1$ if $x\in X$ and $\indicator{X}(x)=0$ otherwise.

\begin{theorem}[Strong $\subseteq$ \unaware{}, $\gpsc\subseteq\gpuc$]\label{th:strong_to_unaware}
$f\in\gpsc$ iff  there exist polynomials $\Upsilon,\Lambda,\Theta$
and a constant polynomial\footnote{$\Omega(x)=c$ for all $x$ for some constant $c$.} $\Omega$ such that $f\in\guc{\Upsilon}{\Omega}{\Lambda}{\Theta}$.
\end{theorem}

where

\begin{definition}[\Unaware{} computability]\label{def:guc}
Let $n,m\in\N$, $f:\subseteq\R^n\rightarrow\R^m$, $\Upsilon:\Rp^3\rightarrow\Rp$
and $\Omega,\Lambda,\Theta:\Rp^2\rightarrow\Rp$.
We say that $f$ is $(\Upsilon,\Omega,\Lambda,\Theta)$-\unawarely{}-computable
if and only if there exist $\delta\geqslant0$, $d\in\N$ and
$(g:\R^{d}\times\R^{n+1}\rightarrow\R^d)\in\gpval[\K]{}$  such that for any
$x\in C^0(\Rp,\R^n)$, $\mu\in C^0(\Rp,\Rp)$, $y_0\in\R^d$, $e\in C^0(\Rp,\R^d)$
there exists (a unique) $y:\Rp\rightarrow\R^d$ satisfying for all $t\in\Rp$:
\begin{itemize}
\item $y(0)=y_0$ and $y'(t)=g(t,y(t),x(t),\mu(t))+e(t)$
\item $\infnorm{y(t)}\leqslant
\Upsilon\left(\pastsup{\delta}{\infnorm{x}}(t),\pastsup{\delta}{\mu}(t),
    \infnorm{y_0}\indicator{[1,\delta]}(t)+\int_{\max(0,t-\delta)}^t\infnorm{e(u)}du\right)$
\item For any $I=[a,b]$, if there exist $\bar{x}\in\dom{f}$ and $\check{\mu},\hat{\mu}\geqslant0$ such that
for all $t\in I$:
\[\mu(t)\in[\check{\mu},\hat{\mu}]\text{ and }\infnorm{x(t)-\bar{x}}\leqslant e^{-\Lambda(\infnorm{\bar{x}},\hat{\mu})}
\text{ and }\int_{a}^b\infnorm{e(u)}du\leqslant e^{-\Theta(\infnorm{\bar{x}},\hat{\mu})}\]
then
\[\infnorm{y_{1..m}(u)-f(\bar{x})}\leqslant e^{-\check{\mu}}\text{ whenever }
a+\Omega(\infnorm{\bar{x}},\hat{\mu})\leqslant u\leqslant b.\]
\end{itemize}
We denote by $\guc{\Upsilon}{\Omega}{\Lambda}{\Theta}$ the set of
$(\Upsilon,\Omega,\Lambda,\Theta)$-\unawarely{}-computable functions
and by $\gpuc$ the set of
$(\poly,\poly,\poly,\poly)$-\unawarely{}-computable functions.
\end{definition}

Actually we prove the implication from left to right. The equivalence
will follow from other sections.

\subsection{Some remarks}

A very common pattern in signal processing is known as ``sample and hold'',
where we have a variable signal and we would like to apply some process to it. Unfortunately,
the processor often assumes (almost) constant input and does not work in real time (analog-to-digital
converters are a typical example).
In this case, we cannot feed the signal directly to the processor so we need some black
box that samples the signal to capture its value, and holds this value long enough
for the processor to compute its output. This process is usually used in a $\tau$-periodic
fashion: the box samples for time $\delta$ and holds for time
$\tau-\delta$. We will need two intermediate lemmas before introducing sample and hold.

\begin{lemma}[``low-X-high'' and ``high-X-low'', \citep{BGP16Arxiv}]\label{lem:lxh_hxl}
For every $I=[a,b]$, there exists $\lxh_I,\hxl_I\in\gpval$ such that for every $\mu\in\Rp$ and $t,x\in\R$
we have:
\begin{itemize}
\item $\lxh_I$ is of the form $\lxh_I(t,\mu,x)=\phi_1(t,\mu,x)x$ where $\phi_1\in\gpval$,
\item $\hxl_I$ is of the form $\lxh_I(t,\mu,x)=\phi_2(t,\mu,x)x$ where $\phi_2\in\gpval$,
\item if $t\leqslant a, |\lxh_I(t,\mu,x)|\leqslant e^{-\mu}$ and $|x-\hxl_I(t,\mu,x)|\leqslant e^{-\mu}$,
\item if $t\geqslant b, |x-\lxh_I(t,\mu,x)|\leqslant e^{-\mu}$ and $|\hxl_I(t,\mu,x)|\leqslant e^{-\mu}$,
\item in all cases, $|\lxh_I(t,\mu,x)|\leqslant|x|$ and $|\hxl_I(t,\mu,x)|\leqslant|x|$.
\end{itemize}
\end{lemma}

\begin{lemma}[``periodic low-integral-low'']\label{lem:plil}
There is a family of functions $\plil_{I,\tau}\in\gpval$  where
$\mu,\tau\in\Rp$, $I=[a,b]\subsetneq[0,\tau]$ and $x\in\R$ with the following property: 
there exist a constant $K$ and $\phi$ such that $\plil_{I,\tau}(t,\mu,x)=\phi(t,\mu,x)x$ and:
\begin{itemize}
\item $\plil_{I,\tau}(\cdot,\mu,x)$ is $\tau$-periodic
\item for all $t\notin I$, $|\plil_{I,\tau}(t,\mu,x)|<e^{-\mu}$
\item for any $\alpha:I\rightarrow\Rp,\beta:I\rightarrow\R$:
\[1\leqslant\int_a^b \phi(t,\alpha(t),\beta(t))dt\leqslant K\]
\end{itemize}
\end{lemma}

\begin{definition}[``periodic low-integral-low'']
Let $t\in\R,\tau\in\Rp,\mu,x\in\R,I=[a,b]\subseteq[0,\tau]$ with $0<b-a<\tau$ and define:
\[\plil_{I,\tau}(t,\mu,x)=\lxh_J(f(t),\nu,K)x\]
where
\[\delta=b-a
\qquad \omega=\frac{2\pi}{\tau}
\qquad K=\frac{1}{4}+\frac{2}{\delta}
\qquad t_1=\frac{a+b}{2}-\frac{\tau}{4}\]
\[\nu=\mu+2+\ln(1+x^2)
\qquad f(t)=\sin(\omega(t-t_1))
\quad J=\left[f(a),f\left(a+\frac{\delta}{4}\right)\right]\]
\end{definition}

\begin{proof}[of \lemref{lem:plil}]
The $\tau$-periodicity is trivial. Using trigonometric identities, observe that
\[f(t)-f(a)=-2\sin\left(\omega\frac{t-b}{2}\right)\sin\left(\omega\frac{t-a}{2}\right)\]
Now it is easy to see that if $t\in[0,a]$ then
$\omega\frac{t-b}{2},\omega\frac{t-a}{2}\in[-\pi,0]$ thus $f(t)\leqslant f(a)$.
By the choice of $J$ and \lemref{lem:lxh_hxl}, we get that
$\lxh_J(f(t),\mu+2,K)\leqslant e^{-\nu}$.
Similarly if $t\in[b,\tau]$ then $\omega\frac{t-b}{2},\omega\frac{t-a}{2}\in[0,\pi]$
and we get the same result. We conclude the first part of the result using that
$|xe^{-\nu}|\leqslant e^{-\mu}$.

Let $\alpha:I\rightarrow\Rp,\beta:I\rightarrow\R$.
Let $a'=a+\frac{\delta}{4}$ and $b'=b-\frac{\delta}{4}$.
Since $\lxh>0$, we have $\int_a^b\plil_{I,\tau}(t,\alpha(t),\beta(t))dt\geqslant\int_{a'}^{b'}\plil_{I,\tau}(t,\alpha(t),\beta(t))dt$.
Again observe that
\[f(t)-f(a')=-2\sin\left(\omega\frac{t-b'}{2}\right)\sin\left(\omega\frac{t-a'}{2}\right)\]
Consequently, if $t\in[a',b']$ then $f(t)\geqslant f(a')$.
By the choice of $J$ and \lemref{lem:lxh_hxl}, we get that
$\lxh_J(f(t),\nu,K)\geqslant K-e^{-\nu}\geqslant K-\frac{1}{4}$ since $\nu\geqslant2$.
Finally $\int_a^b\plil_{I,\tau}(t,\alpha(t),\beta(t))dt\geqslant (b'-a')(K-\frac{1}{4})\geqslant1$
and $\int_a^b\plil_{I,\tau}(t,\alpha(t),\beta(t))dt\leqslant (b-a)K$ by \lemref{lem:lxh_hxl}.


Apply \lemref{lem:gpac_ext_class_stable} multiple times to get that $\plil_{I,\tau}\in\gval{\poly}$.
\end{proof}

\begin{lemma}[Sample and hold]\label{lem:sample}
There is a family of functions  $\sample_{I,\tau}(t,\mu,x,g) \in \gpval$, where
$t\in\R,\mu,\tau\in\Rp,x,g\in\R,I=[a,b]\subsetneq[0,\tau]$, with the following property: 
let $\tau\in\Rp$, $I=[a,b]\subsetneq[0,\tau]$, $y:\Rp\rightarrow\R$,
$y_0\in\R$, $x,e\in C^0(\Rp,\R)$ and
$\mu:\Rp\rightarrow\Rp$ be an increasing function. Suppose that for all $t\in\Rp$:
\[y(0)=y_0\qquad y'(t)=\sample_{I,\tau}(t,\mu(t),y(t),x(t))+e(t)\]
Then:
\[|y(t)|\leqslant2+\smashoperator{\int_{\max(0,t-\tau-|I|)}^t}|e(u)|du+\max\left(|y(0)|\indicator{[0,b]}(t),\pastsup{\tau+|I|}|x|(t)\right)\]
Furthermore:
\begin{itemize}
\item if $t\notin I\pmod{\tau}$ then $|y'(t)|\leqslant e^{-\mu(t)}+|e(t)|$
\item for $n\in\N$, if there exist $\bar{x}\in\R$ and $\nu,\nu'\in\Rp$ such that
$|\bar{x}-x(t)|\leqslant e^{-\nu}$ and $\mu(t)\geqslant\nu'$ for
all $t\in n\tau+I$ then
\[|y(n\tau+b)-\bar{x}|\leqslant\int_{n\tau+I}|e(u)|du+e^{-\nu}+e^{-\nu'}.\]
\item for $n\in\N$, if there exist $\check{x},\hat{x}\in\R$ and $\nu\in\Rp$ such that $x(t)\in[\check{x},\hat{x}]$ and $\mu(t)\geqslant\nu$ for
all $t\in n\tau+I$ then
\[y(n\tau+b)\in[\check{x}-\varepsilon,\hat{x}+\varepsilon]\]
where $\varepsilon=2e^{-\nu}+\int_{n\tau+I}|e(u)|du$.
\item for any $J=[c,d]\subseteq\Rp$, if there exist $\nu,\nu'\in\Rp$ and $\bar{x}\in\R$ such that
$\mu(t)\geqslant\nu'$ for all $t\in J$ and $|x(t)-\bar{x}|\leqslant e^{-\nu}$ for all $t\in J\cap(n\tau+I)$ for some $n\in\N$,
then
\[|y(t)-\bar{x}|\leqslant e^{-\nu}+e^{-\nu'}+\int_{t-\tau-|I|}^t|e(u)|du\] for all $t\in[c+\tau+|I|,d]$.
\item if there exists $\Omega:\Rp\rightarrow\Rp$ such that for any $J=[a,b]$ and $\bar{x}\in\R$ such that for all
$\nu\in\Rp$, $n\in\N$ and $t\in(n\tau+I)\cap[a+\Omega(\nu),b]$ we have $|\bar{x}-x(t)|\leqslant e^{-\nu}$, then
\[|y(t)-\bar{x}|\leqslant e^{-\nu}\] for all $t\in[a+\Omega^*(\nu),b]$ where
\[\Omega^*(\nu)=\max(\Omega(\nu+\ln3),\mu^{-1}(\nu+\ln3))+\tau+|I|.\]
\end{itemize}
\end{lemma}

\begin{definition}[Sample and hold]
Let $t\in\R,\mu,\tau\in\Rp,x,g\in\R,I=[a,b]\subsetneq[0,\tau]$ and define:
\[\sample_{I,\tau}(t,\mu,x,g)=\plil_{I,\tau}(t,\hat{\mu},\reach(\check{\mu},x,g))\]
where
\[\check{\mu}=\frac{\mu+1}{\min(1,|I|)}\qquad\hat{\mu}=\mu+\max(0,\ln(\tau-|I|))\]
\end{definition}

\begin{proof}
Let $n\in\N$. Apply
\lemref{lem:plil}, \lemref{lem:reach} and Remark \ref{rem:reach_mult} to get that:
\begin{itemize}
\item For all $t\in I_n=[n\tau+a,n\tau+b]$:
\[y'(t)=\phi(t)\reach(\check{\mu}(t),y(t),x(t))+e(t)\] where $\int_{I_n}\phi\geqslant1$.
Since $|x(t)-0|\leqslant\sup_{u\in I_n}|x(u)|$ and
\[\int_{I_n}\phi\check{\mu}=\int_{I_n}\phi\frac{1+\mu}{|I|}\geqslant1\]
then
\begin{align*}
|y(n\tau+b)-0|
&\leqslant \sup_{I_n}|x(u)|+\int_{I_n}|e(u)|du+e^{-1}\\
&\leqslant1+\sup_{u\in I_n}|x(u)|+\int_{I_n}|e(u)|du.
\end{align*}
\item For all $t\in [n\tau+b,(n+1)\tau+a]$:
\[|y'(t)|\leqslant|e(t)|+e^{-\hat{\mu}(t)}\leqslant|e(t)|+e^{-\ln(\tau-|I|)}\]
thus
\begin{align*}
|y(t)-0|
&\leqslant\int_{n\tau+b}^{t}|e(u)|du+(\tau-|I|)e^{-\ln(\tau-|I|)}\\
&\qquad+1+\sup_{u\in I_n}|x(u)|+\int_{I_n}|e(u)|du\\
&\leqslant2+\sup_{u\in I_n}|x(u)|+\int_{n\tau+a}^{t}|e(u)|du.
\end{align*}
\item For all $t\in I_{n+1}$:
\[y'(t)=\reach(\phi(t)\check{\mu}(t),y(t),x(t))\] where $\int_{I_n}\phi\geqslant1$.
Since $|x(t)-0|\leqslant\sup_{u\in I_{n+1}}|x(u)|$
then
\begin{align*}
|y(t)-0|
&\leqslant\max\left(\sup_{u\in[(n+1)\tau+a,t]}|x(u)|,|y((n+1)\tau+a)-0|\right)\\
&\qquad+\int_{(n+1)\tau+a}^t|e|\\
&\leqslant2+\sup_{u\in[n\tau+a,t]}|x(u)|+\int_{n\tau+a}^{t}|e(u)|du.
\end{align*}
\end{itemize}
Note that this analysis is a bit subtle: the first point \emph{does not} give a bound on $|y(t)|$
over $I_n$, it only gives a bound on $|y(n\tau+b)|$. On the contrary the two other points give
bounds on $|y(t)|$ over $[n\tau+b,(n+1)\tau+b]$ which cover the whole period so
by correctly putting everything together, we get that for all
$|y(t)|\leqslant2+\sup_{u\in[t,t-\tau-|I|]\cap\Rp}$ $|x(u)|+\int_{t-\tau-|I|}^t|e(u)|du$
for all $t\geqslant b$. The case of the initial segment is similar in aspect but uses the
other result from \lemref{lem:reach}:
\begin{itemize}
\item For all $t\in [0,a]$:
\[|y'(t)|\leqslant|e(t)|+e^{-\hat{\mu}(t)}\leqslant|e(t)|+e^{-\ln(\tau-|I|)}\]
thus
\[|y(t)|\leqslant\int_0^t|e(u)|du+ae^{-\ln(\tau-|I|)}+|y_0|\leqslant\int_0^t|e(u)|du+1+|y_0|.\]
\item For all $t\in [a,b]$:
\[y'(t)=\reach(\phi(t)\check{\mu}(t),y(t),x(t))+e(t)\] where $\int_{I_n}\phi\geqslant1$.
Since $|x(t)-0|\leqslant\sup_{u\in[a,t]}|x(u)|$
then
\begin{align*}
|y(t)-0|&\leqslant\max(\sup_{u\in[a,t]}|x(u)|,|y(a)-0|)+\int_a^t|e(u)|du\\
&\leqslant1+\int_0^t|e(u)|du+\max(|y_0|,\sup_{u\in[a,t]}|x(u)|).
\end{align*}
\end{itemize}
Finally, we get that for all $t\in\Rp$:
\[|y(t)|\leqslant2+\smashoperator{\int_{t-\tau-|I|}^t}|e(u)|du+\max\left(|y(0)|\indicator{[0,b]}(t),\pastsup{\tau+|I|}|x|(t)\right)\]

The first extra statement is a trivial consequence of \lemref{lem:plil} and the
fact that $\check{\mu}(t)\geqslant\mu(t)$.

The second extra statement has mostly been proved already and uses \lemref{lem:plil} and \lemref{lem:reach} again.
Let $n\in\N$, assume there exist $\bar{x}\in\R$ and $\nu\in\Rp$ such as described.
For all $t\in I_n=[n\tau+a,n\tau+b]$ we have
\[y'(t)=\phi(t)\reach(\check{\mu}(t),y(t),x(t))+e(t)\] where $\int_{I_n}\phi\geqslant1$.
Since $|x(t)-\bar{x}|\leqslant e^{-\nu}$ and $\int_{I_n}\phi\check{\mu}=\int_{I_n}\phi\frac{1+\mu}{|I|}\geqslant\nu'$
then
\[|y(n\tau+b)-\bar{x}|\leqslant e^{-\nu}+\int_{I_n}|e(u)|du+e^{-\nu'}.\]

The third statement is a consequence of the previous one: since $n\tau+I$ is a compact set
and $x$ is a continuous function, it admits a maximum over $n\tau+I$. Apply the
previous statement to $\frac{\bar{x}+\sup_{n\tau+I}x}{2}\geqslant\bar{x}$ to conclude.


The last extra statement requires more work.
Let $\nu\geqslant0$ and $n\in\N$ such that $n\tau+a\geqslant\Omega(\nu)$. Apply
\lemref{lem:plil}, Remark \ref{rem:reach_mult} and \lemref{lem:reach} to get that:
\begin{itemize}
\item For all $t\in I_n$:
\[y'(t)=\phi(t)\reach(\check{\mu}(t),y(t),x(t))\]
where $\int_{I_n}\phi\geqslant1$.
Since $t\geqslant n\tau+a\geqslant\Omega(\nu)$ and $t\in I_n$ then $|x(t)-\bar{x}|\leqslant e^{-\nu}$.
And since
\[\int_{I_n}\phi\check{\mu}=\int_{I_n}\phi\frac{1+\mu}{|I|}\geqslant1+\mu(n\tau+a)\]
then
\[|y(n\tau+b)-\bar{x}|\leqslant e^{-\nu}+e^{-\mu(n\tau+a)}.\]
\item For all $t\in [n\tau+b,(n+1)\tau+a]$:
\[|y'(t)|\leqslant e^{-\hat{\mu}(t)}\leqslant e^{-\hat{\mu}(n\tau+a)}\]
thus
\begin{align*}
|y(t)-\bar{x}|
&\leqslant(\tau-|I|)e^{-\hat{\mu}(n\tau+a)}+e^{-\nu}+e^{-\mu(n\tau+a)}\\
&\leqslant e^{-\nu}+2e^{-\mu(n\tau+a)}.
\end{align*}
\item For all $t\in I_{n+1}$:
\[y'(t)=\phi(t)\reach(\check{\mu}(t),y(t),x(t))\]
where $\int_{I_n}\phi\geqslant1$.
Since $t\geqslant n\tau+a\geqslant\Omega(\nu)$ and $t\in I_n$ then $|x(t)-\bar{x}|\leqslant e^{-\nu}$.
Thus
\[|y(t)-\bar{x}|\leqslant\max(e^{-\nu},|y((n+1)\tau+a)-\bar{x}|)\leqslant e^{-\nu}+2e^{-\mu(n\tau+a)}.\]
\end{itemize}
Finally, we get that
\[|y(t)-\bar{x}|\leqslant e^{-\nu}+2e^{-\mu(n\tau+a)}\]
for all $t\in[n\tau+b,(n+1)\tau+b]$.

Define
\[\Omega^*(\nu)=\max(\Omega(\nu+\ln3),\mu^{-1}(\nu+\ln3))+\tau+|I|.\]
Let $\nu\geqslant0$
and $t\geqslant \Omega^*(\nu)$. Let $n\in\N$ such that $t\in[n\tau+b,(n+1)\tau+b]$.
Then
\begin{align*}
n\tau+a
&=(n+1)\tau+b-\tau-|I|\\
&\geqslant t-\tau-|I|\\
&\geqslant \Omega^*(\nu)-\tau-|I|\\
&\geqslant\Omega(\nu+\ln3).
\end{align*}
By the previous reasoning, we get that $|y(t)-\bar{x}|\leqslant e^{-\nu}+2e^{-\mu(n\tau+a)}$.
And since
\begin{align*}
n\tau+a
&\geqslant\Omega^*(\nu)-\tau-|I|\\
&\geqslant\mu^{-1}(\nu+\ln3)
\end{align*}
then $\mu(n\tau+a)\geqslant\nu+\ln3$.
Thus $|y(t)-\bar{x}|\leqslant3e^{-\nu}\leqslant e^{-\nu}$.
\end{proof}

\subsection{The proof}

We then get to the proof of \thref{th:strong_to_unaware}

\begin{proof}
Let $(f:\subseteq\R^n\rightarrow\R^m)\in\gsc{\Upsilon}{\Omega}{\Theta}$ where
$\Upsilon$, $\Omega$ $\Theta$ are polynomials which we assume, without loss of generality,
to be increasing functions of theirs inputs. Apply \defref{def:gsc} to get
$d$, $h$ and $g$.

Let $e=1+d+m$, $x\in C^0(\Rp,\R^n)$, $\mu\in C^0(\Rp,\Rp)$, $(\nu_0,y_0,z_0)\in\R^e$,
$(e_\nu,e_y,e_z)\in C^0(\Rp,\R^e)$
and consider the following system:
\[
\left\{\begin{array}{@{}r@{}l}
\nu(0)&=\nu_0\\y(0)&=y_0\\z(0)&=z_0
\end{array}\right.\quad
\left\{\begin{array}{@{}r@{}l}
\nu'(t)&=\sample_{[0,1],4}(t,\mu^*(t),\nu(t),\mu(t)+\ln\Delta+7)+e_\nu(t)\\
y'(t)&=\sample_{[1,2],4}(t,\mu^*(t),y(t),g(x(t),\nu(t)))\\
        &\hspace{1em}+\plil_{[2,3],4}(t,\mu^*(t),A(t)h(y(t)))+e_y(t)\\
z'(t)&=\sample_{[3,4],4}(t,\mu^*(t),z(t),y_{1..m}(t))+e_z(t)
\end{array}\right.
\]
where
\[\Delta=5\qquad\Delta'=\ln\Delta+10\]
\[\mu^*(t)=\mho^*(1+\norm_{\infty,1}(x(t)),\nu(t)+4)\]
\[A(t)=1+\Omega(1+\norm_{\infty,1}(x(t)),\nu(t))\]
\[\Lambda^*(\alpha,\mu)=\Theta^*(\alpha,\mu)=\mho^*(\alpha,\mu+\Delta')\]
\[\mho^*(\alpha,\mu)=\mu+\ln\Delta+\Theta(\alpha,\mu)+\ln q(\alpha+\mu)\]

Let $I=[a,b]$ and assume there exist $\bar{x}\in\dom{f}$ and $\check{\mu},\hat{\mu}\in\Rp$ such that
for all $t\in I$, $\mu(t)\in[\check{\mu},\hat{\mu}]$, $\infnorm{x(t)-\bar{x}}\leqslant e^{-\Lambda^*(\infnorm{\bar{x}},\hat{\mu})}$
and $\int_a^b\infnorm{e(u)}du\leqslant e^{-\Theta^*(\infnorm{\bar{x}},\hat{\mu})}$.
Apply \thref{prop:generable_mod_cont} to $g$ to get $q\in\K[\R]$, without loss of generality
we can assume that $q$ is an increasing function and $q\geqslant1$.
We will use \lemref{lem:norm} to get that $\norm_{\infty,1}(x(t))+1\geqslant\infnorm{\bar{x}}$
because $\infnorm{x(t)-\bar{x}}\leqslant1$. Also note that $\mu^*,\Theta^*,\Lambda^*$
are increasing functions of their arguments.
Let $n\in\N$ such that $[4n,4n+4]\subseteq I$ and $t\in[4n,4n+4]$. We will first analyse the variable $\nu$,
note that the analysis is extremely rough to simplify the proof.
\begin{itemize}
\item \textbf{if $t\in[4n,4n+1]$ then} $\mu^*(t)\geqslant0$ so apply
  \lemref{lem:sample} to get that
\[\nu(4n+1)\in[\check{\mu}+\ln\Delta+7-\varepsilon,\hat{\mu}+\ln\Delta+7+\varepsilon]\]
where
\[\varepsilon\leqslant2e^{-0}+\int_{4n}^{4n+1}|e_\nu(u)|du\leqslant3\]
because $\int_a^b\infnorm{e(t)}dt\leqslant1$. Define $\bar{\nu}=\nu(4n+1)$, then
\[\bar{\nu}\in[\check{\mu}+\ln\Delta+4,\hat{\mu}+\underbrace{\ln\Delta+10}_{=\Delta'}].\]
\item \textbf{if $t\in[4n+1,4n+4]$ then} $\mu^*(t)\geqslant0$ so apply
  \lemref{lem:sample} to get that
\[|\nu'(t)|\leqslant e^{-0}+\int_{4n+1}^t|e_\nu(u)|du\]
and thus
\[|\nu(t)-\bar{\nu}|\leqslant(t-4n-1)+\int_{4n+1}^t\infnorm{e(u)}du\leqslant4\]
because $\int_a^b\infnorm{e(t)}dt\leqslant1$. In other words
\[\nu(t)\in[\bar{\nu}-4,\bar{\nu}+4].\]
\end{itemize}
Furthermore for $t\in[4n+1,4n+4]$ we have:
\[
\mu^*(t)\geqslant\Theta^*(1+\norm_{\infty,1}(x(t)),\nu(t)+4)
    \geqslant\mho^*(\infnorm{\bar{x}},\bar{\nu})
\]
It will also be useful to note that:
\begin{align*}
\Lambda^*(\infnorm{\bar{x}},\hat{\mu})
    =\Theta^*(\infnorm{\bar{x}},\hat{\mu})
    &\geqslant\mho^*(\infnorm{\bar{x}},\hat{\mu}+\Delta')\\
    &\geqslant\mho^*(\infnorm{\bar{x}},\bar{\nu})
\end{align*}
We can now analyse $y$ using this property:
\begin{itemize}
\item \textbf{if $t\in[4n+1,4n+2]$ then}
\[|\nu'(t)|\leqslant e^{-\mu^*(t)}+|e_\nu(t)|\]
thus
\[|\nu(t)-\bar{\nu}|\leqslant e^{-\mho^*(\infnorm{\bar{x}},\bar{\nu})}+\int_{4n+1}^{4n+2}|e_\nu(u)|du.\]
Furthermore
\[\sup_{[4n+1,4n+2]}\infnorm{x}\leqslant\infnorm{\bar{x}}+1,\]
thus:
\begin{align*}
& \infnorm{g(\bar{x},\bar{\nu})-g(x(t),\nu(t))} \\
&\qquad \qquad \leqslant\max(|\nu(t)-\bar{\nu}|,\infnorm{x(t)-\bar{x}})q(\max(\infnorm{\bar{x}},|\bar{\nu}|))\\
&\qquad \qquad \leqslant\max\left(e^{-\Theta^*(\infnorm{\bar{x}},\hat{\mu})}+
    e^{-\mho^*(\infnorm{\bar{x}},\bar{\nu})},e^{-\Lambda^*(\infnorm{\bar{x}},\hat{\mu})}\right)q(\infnorm{\bar{x}}+\bar{\nu})\\
&\qquad \qquad  \leqslant 2e^{-\Theta(\infnorm{\bar{x}},\bar{\nu})-\ln\Delta}
\end{align*}
Also note that
\[\infnorm{y'(t)-\sample_{[1,2],4}(t,\mu^*(t),y(t),g(x(t),\nu(t)))}\leqslant e^{-\mu^*(t)}\]
by \lemref{lem:plil}. So we can apply \lemref{lem:sample} to get that
\begin{align*}
\infnorm{y(4n+2)-g(\bar{x},\bar{\nu})}
&\leqslant2e^{-\Theta(\infnorm{\bar{x}},\bar{\nu})-\ln\Delta}+e^{-\mho^*(\infnorm{\bar{x}},\bar{\nu})}\\
&\qquad+\int_{4n+1}^{4n+2}\infnorm{e(u)}du\\
&\leqslant 4e^{-\Theta(\infnorm{\bar{x}},\bar{\nu})-\ln\Delta}.
\end{align*}
\item \textbf{if $t\in[4n+2,4n+3]$ then} apply Lemmas \ref{lem:sample} and \ref{lem:plil}
to get $\phi$ such that $\int_{4n+2}^{4n+3}\phi(u)du\geqslant1$ and
\[\infnorm{y'(t)-\phi(t)A(t)h(y(t))}\leqslant e^{-\mu^*(t)}+\infnorm{e_y(t)}.\]
Define $\psi(t)=\int_{4n+2}^t\phi(u)A(u)du$ then
$\psi(4n+3)\geqslant\Omega(\infnorm{\bar{x}},\bar{\nu})$ since
$A(u)\geqslant\Omega(\infnorm{\bar{x}},\bar{\nu})$ for $u\in[4n+2,4n+3]$.
Apply \lemref{lem:perturbed_time_scale_gpac} over $[4n+2,4n+3]$ to get that
$y(t)=w(\psi(t))$ where $w$ satisfies
\[w(0)=y(4n+2),\qquad w'(\xi)=h(w(\xi))+\tilde{e}(\xi)\]
where $\tilde{e}\in C^0(\Rp,\R^d)$ satisfies
\[\int_{0}^{\psi(t)}\infnorm{\tilde{e}(\xi)}d\xi=\int_{4n+2}^t\infnorm{e_y(u)}du
\leqslant e^{-\Theta^*(\infnorm{\bar{x}},\hat{\mu})}\leqslant e^{-\Theta(\infnorm{\bar{x}},\bar{\nu})-\ln\Delta}.\]
Hence $\infnorm{w(0)-g(\bar{x},\bar{\nu})}\leqslant4e^{-\Theta(\infnorm{\bar{x}},\bar{\nu})-\ln\Delta}$
from the result above. In other words:
\[w(0)=g(\bar{x},\bar{\nu})+\tilde{e}_0,\qquad
w'(t)=g(w(t))+\tilde{e}(t)\]
where
\[\infnorm{\tilde{e}_0}+\int_0^{\psi(t)}\infnorm{e(u)}du\leqslant 5e^{-\Theta(\infnorm{\bar{x}},\bar{\nu})-\ln\Delta}
\leqslant e^{-\Theta(\infnorm{\bar{x}},\bar{\nu})}\]
because $\Delta\geqslant5$.
Apply \defref{def:gsc} to get that
\[\infnorm{w_{1..m}(\psi(4n+3))-f(\bar{x})}\leqslant e^{-\bar{\nu}}\]
since $\psi(4n+3)\geqslant\Omega(\infnorm{\bar{x}},\bar{\nu})$.
\item \textbf{if $t\in[4n+3,4n+4]$ then} $\infnorm{y'(t)}\leqslant e^{-\mu^*(t)}+\infnorm{e_y(t)}$
thus
\begin{align*}
\infnorm{y(t)-y(4n+3)}
&\leqslant e^{-\mho^*(\infnorm{\bar{x}},\bar{\nu})}+\int_{4n+3}^t\infnorm{e_y(u)}du\\
&\leqslant 2e^{-\bar{\nu}}
\end{align*}
so $\infnorm{y_{1..m}(t)-f(\bar{x})}\leqslant3e^{-\bar{\nu}}$.
\end{itemize}
Note that the above reasoning is also true for the last segment $[4n,b]\subseteq I$
in which case the result only applies up to time $b$ of course. In other words,
the results apply as long as $t\in[4n,4+4]\cap I$ and $4n\geqslant a$.
From this we conclude that if $t\in[a+4,b]\cap[4n+3,4n+3]$ for some $n\in\N$ then
$\infnorm{y_{1..m}(t)-f(\bar{x})}\leqslant3e^{-\bar{\nu}}$.
Apply \lemref{lem:sample} to get, using that $\bar{\nu}\geqslant\check{\mu}+\ln\Delta$ and $\Delta\geqslant5$,
that for all $t\in[a+5,b]$:
\begin{align*}
\infnorm{z(t)-f(\bar{x})}
    &\leqslant 3e^{-\bar{\nu}}+e^{-\mho^*(\infnorm{\bar{x}},\bar{\nu})}+\int_{t-5}^t\infnorm{e(u)}du
    \leqslant 5e^{-\bar{\nu}}\\
    &\leqslant e^{-\check{\mu}}
\end{align*}

To complete the proof, we must also analyse the norm of the system. As a shorthand, we introduce
the following notation:
\[\operatorname{int}_\delta^+\alpha(t)=\int_{\max(0,t-\delta)}^t\alpha(u)du\]
Apply \lemref{lem:sample}
to get that:
\begin{align*}
|\nu(t)|&\leqslant 2+\smashoperator{\int_{\max(0,t-5)}^t}|e_\nu(u)|du
        +\max\left(|\nu_0|\indicator{[0,4]}(t),\pastsup{5}{|\mu+\ln\Delta+7|}(t)\right)\\
        &\leqslant\poly\left(|\nu_0|\indicator{[0,5]}(t)+\operatorname{int}_{5}^+|e_\nu|(t),\pastsup{5}\mu(t)\right)
\end{align*}
The analysis of $y$ is a bit more painful, as it uses both results about the
sampling function and the strongly-robust system we are simulating. Let $n\in\N$,
and $t\in[4n,4n+4]$:
\begin{itemize}
\item \textbf{if $t\in[4n,4n+1]$ then} apply Lemmas \ref{lem:sample} and \ref{lem:plil}
to get, using that $\mu(t)\geqslant0$, that $\infnorm{y'(t)}\leqslant2+\infnorm{e(t)}$ and thus
$\infnorm{y(t)-y(4n)}\leqslant2+\int_{4n}^t\infnorm{e(u)}du$.
\item \textbf{if $t\in[4n+1,4n+2]$ then} using the result on $\nu$, we have
\begin{multline}\infnorm{g(x(t),\nu(t))}\leqslant\sup_{[4n+1,t]}\poly(\infnorm{x},\nu) \\
\leqslant
\poly\left(|\nu_0|\indicator{[0,5]}(t)+\operatorname{int}_{6}^+\infnorm{e}(t),\pastsup{6}\mu(t),\pastsup{1}\infnorm{x}(t)\right).
\end{multline}
Apply Lemmas \ref{lem:sample} and \ref{lem:plil}
to get, using that $\mu(t)\geqslant0$ and the result on $\nu$, that:
\begin{multline}
\infnorm{y(4n+2)} \leqslant\sup_{[4n+1,4n+2]}\infnorm{g(x,\nu)}+2+\int_{4n+1}^{4n+2}\infnorm{e(u)}du
    \\ \leqslant\poly\left(|\nu_0|\indicator{[0,5]}(4n+2)+\operatorname{int}_{6}^+\infnorm{e}(4n+2),\pastsup{6}\mu(4n+2),\right.\\
               \left.          \pastsup{1}\infnorm{x}(4n+2)\right) \\
\end{multline}
and also that:
\begin{align*}
\infnorm{y(t)}&\leqslant\max\left(\sup_{[4n+1,t]}\infnorm{g(x,\nu)}+2,\infnorm{y(4n+1)}\right)+\int_{4n+1}^t\infnorm{e(u)}du\\
    &\leqslant\poly\left(|\nu_0|\indicator{[0,5]}(t)+\operatorname{int}_{6}^+\infnorm{e}(t),\pastsup{6}\mu(t),
                         \pastsup{1}\infnorm{x}(t),\infnorm{y(4n)}\right)
\end{align*}
\item \textbf{if $t\in[4n+2,4n+3]$ then} apply \lemref{lem:sample},
  Lemmas \ref{lem:plil},
\ref{lem:perturbed_time_scale_gpac} and \ref{def:gsc} to get that
\[\infnorm{y(t)}\leqslant \Upsilon(0,0,\hat{e}(\hat{A}(t)),\hat{A}(t))\]
where $\hat{A}(t)=\int_{4n+2}^tA(u)du$ and
\[\hat{e}(\hat{A}(t))=\infnorm{y(4n+2)-g(0,0)}+\int_{4n+2}^t1+\infnorm{e(u)}du.\]
Since $\Omega$ is a polynomial, and using the result on $\nu$, we get that:
\begin{align*}
\hat{A}(t)&\leqslant\sup_{[4n+2,t]}\poly(\infnorm{x},|\nu|)\\
    &\leqslant\poly\left(|\nu_0|\indicator{[0,5]}(t)+\operatorname{int}_{6}^+\infnorm{e},\pastsup{6}\mu(t),\pastsup{1}{\infnorm{x}}(t)\right)
\end{align*}
and using that $4n+2\leqslant t\leqslant4n+3$:
\begin{align*}
\infnorm{y(4n+2)-g(0,0)}&\leqslant\infnorm{y(4n+2)}+\infnorm{g(0,0)}\\
    &\leqslant\poly\left(|\nu_0|\indicator{[0,5]}(t)+\operatorname{int}_{6}^+\infnorm{e},\pastsup{7}\mu(t),
                         \pastsup{2}\infnorm{x}(t)\right)
\end{align*}
And since $\Upsilon$ is a polynomial, we conclude that:
\[
\infnorm{y(t)}\leqslant\poly\left(|\nu_0|\indicator{[0,5]}(t)+\operatorname{int}_{6}^+\infnorm{e}(t),\pastsup{7}\mu(t),
                         \pastsup{2}\infnorm{x}(t)\right)
\]
\item \textbf{if $t\in[4n+3,4n+4]$ then} apply Lemmas \ref{lem:sample} and \ref{lem:plil}
to get, using that $\mu(t)\geqslant0$, that $\infnorm{y'(t)}\leqslant2+\infnorm{e(t)}$ and thus
\[\infnorm{y(t)-y(4n+3)}\leqslant2+\int_{4n+3}^t\infnorm{e(u)}du.\]
\end{itemize}
From this analysis we can conclude that for all $t\in[0,2]$:
\begin{align*}
\infnorm{y(t)}&\leqslant\poly\left(|\nu_0|\indicator{[0,5]}(t)+\operatorname{int}_{6}^+\infnorm{e}(t),\pastsup{6}\mu(t),
                         \pastsup{1}\infnorm{x}(t),\infnorm{y(0)}\right)\\
                &\leqslant\poly\left(|\nu_0|+\operatorname{int}_{6}^+\infnorm{e}(t),\pastsup{6}\mu(t),
                         \pastsup{1}\infnorm{x}(t),\infnorm{y_0}\right)
\end{align*}
and for all $n\in\N$ and $t\in[4n+2,4n+6]$:
\[
\infnorm{y(t)}\leqslant\poly\left(|\nu_0|\indicator{[0,5]}(t)+\operatorname{int}_{9}^+\infnorm{e}(t),\pastsup{9}\mu(t),
                         \pastsup{4}\infnorm{x}(t)\right)
\]
Putting everything together, we get for all $t\in\Rp$:
\[
\infnorm{y(t)}\leqslant\poly\left(\infnorm{y_0,\nu_0}\indicator{[0,5]}(t)+\operatorname{int}_{9}^+\infnorm{e}(t),\pastsup{9}\mu(t),
                         \pastsup{4}\infnorm{x}(t)\right)
\]
Finally apply \lemref{lem:sample} to get a similar bound on $z$ and thus
on the entire system.
\end{proof}

\section{Proof that AXP implies AOP}
\label{app:unaware_to_online}

We can prove 

\begin{theorem}[\Unaware{} $\subseteq$ online]\label{th:unaware_to_online}
$\gpuc = \gpoc$
\end{theorem}

Actually, we prove in this section that $\gpuc  \subseteq
\gpoc$. Equality will follow from other sections.

\subsection{Some remarks}

We start by the following lemmas:

\begin{lemma}[$\gpuc$ time rescaling]\label{lem:gpuc_time_rescaling}
If $f\in\gpuc$ then there exist polynomials $\Upsilon,\Lambda,\Theta$
and a constant polynomial\footnote{$\Omega(x)=c$ for all $x$ for some constant $c$.} $\Omega$ such that $f\in\guc{\Upsilon}{\Omega}{\Lambda}{\Theta}$.
\end{lemma}

\begin{proof}
We go for the shortest proof: we will show that $\gpuc\subseteq\gpwc$
and use \thref{th:weak_to_robust} then Theorem
\ref{th:robust_to_strong} followed by \thref{th:strong_to_unaware}
which proves exactly our statement.

The proof that $\gpuc\subseteq\gpwc$ is next to trivial since because we are given an \unaware{} system and some input and precision,
we can simply store the input and precision into some variables and feed them into the (extreme) system.
We make the system autonomous by using a variable to store the time.

Let
$(f:\subseteq\R^n\rightarrow\R^m)\in\guc{\Upsilon}{\Omega}{\Lambda}{\Theta}$,
apply \defref{def:guc} to get
$\delta, d$ and $g$.
Let $x\in\dom{f}$ and $\mu\in\Rp$, and consider the following system:
\[
\left\{\begin{array}{@{}r@{}l@{}}
x(0)&=x\\\mu(0)&=\mu\\\tau(0)&=0\\y(0)&=0
\end{array}\right.
\qquad
\left\{\begin{array}{@{}r@{}l@{}}
x'(t)&=0\\\mu'(t)&=0\\\tau'(t)&=1\\
y'(t)&=g(t,y(t),x(t),\mu(t))
\end{array}\right.
\]
Clearly the system is of the form $z(0)=h(x,\mu)$ and $z'(t)=H(z(t))$
where $h$ and $H$ belong to $\gpval$ (and are defined over the entire space).
Apply the definition to get that:
\[\infnorm{y(t)}\leqslant\Upsilon(\infnorm{x},\mu,0)\]
And thus the entire system is bounded by a polynomial in $\infnorm{x},\mu$
and $t$.
Furthermore, if $t\geqslant\Omega(\infnorm{x},\mu)$ then $\infnorm{y_{1..m}(t)-f(x)}\leqslant e^{-\mu}$.
To conclude the proof, we need to rewrite the system as a PIVP using
\thref{prop:gpac_ext_ivp_stable_pre}.
\end{proof}

\subsection{Reaching a value}\label{sec:reach}

The notion of \unaware{} computability might seem so strong at first that one
can wonder if anything is really computable in this sense. In this section,
we will introduce a very useful pattern which we call ``reaching a value''.
This can be seen as a proof that all constant functions or generable functions are \unawarely{}-computable, and this pattern will be used as
a basic block to build more complicated extremely-computable functions.
As as introductory example, consider the system:
\[y'(t)=\alpha-y(t)\]
This system can be shown to converge to $\alpha$ whatever the initial value is.
In this section we extend this system in several non-trivial
ways. In particular, we want to ensure a certain rate of convergence in all
situations and we want to make this system robust to perturbations. In other words,
we want to analyse:
\[y'(t)=\alpha(t)-y(t)+e(t)\]
where $e(t)$ is a perturbation and $\alpha(t)\approx \alpha$.

\begin{definition}[Reach ODE]
Let $T>0$, $I=[0,T]$, $g,E:I\rightarrow\R$, $\phi:I\rightarrow\Rps$.
Define \eqref{eq:reach} as the following differential equation for $t\in I$,
\begin{equation}\label{eq:reach}
\left\{\begin{array}{@{}r@{}l@{}}y'(t)&=\phi(t)X_3(g(t)-y(t))+E(t)\\y(0)&=y_0\end{array}\right.\qquad\text{where }X_3(u)=u+u^3
\end{equation}
\end{definition}

\begin{lemma}[Reach ODE: integral error]\label{lem:reach_ode}
Let $T>0$, $I=[0,T]$, $g,E \in C^0(I,\R)$, $\phi\in C^0(I,\Rps)$.
Assume that there exist $\eta>0$ and $\bar{g}\in\R$ such that
for all $t\in I$ we have $|g(t)-\bar{g}|\leqslant\eta$.
Then the solution $y$ to \eqref{eq:reach} exists over $I$ and satisfies:
\[|y(T)-\bar{g}|\leqslant\eta+\int_0^T|E(t)|dt+\frac{1}{\sqrt{\exp(2\int_0^T\phi(u)du)-1}}\]
Furthermore, for any $t\in I$:
\[|y(t)-\bar{g}|\leqslant\max(\eta,|y(0)-\bar{g}|)+\int_0^t|E(u)|du\]
\end{lemma}

\begin{proof}
Write $f(t,x)=E(t)+\phi(t)X_3(g(t)-x)$, then $y'(t)=f(t,y(t))$. Define $I(t)=\int_{0}^t|E(u)|du$ and consider:
\[f_+(t,x)=|E(t)|+\phi(t)X_3\left(\bar{g}+\eta-(x-I(t))\right)\]
\[f_-(t,x)=-|E(t)|+\phi(t)X_3\left(\bar{g}-\eta-(x+I(t))\right)\]
Since $X_3$ and $I$ are increasing functions, it is easily seen that 
\[f_-(t,x)\leqslant f(t,x)\leqslant f_+(t,x).\]
By a classical result of differential inequalities, we get that
\[y_-(t)\leqslant y(t)\leqslant y_+(t)\]
where $y_-(0)=y_+(0)=y(0)$ and $y_{\pm}'(t)=f_{\pm}(t,y_{\pm}(t))$.
Now realize that:
\[y_+'(t)-I'(t)=\phi(t)X_3(\bar{g}+\eta-(y_+(t)-I(t)))\]
\[y_-'(t)+I'(t)=\phi(t)X_3(\bar{g}-\eta-(y_-(t)+I(t)))\]
which are two instances of the following differential equation:
\[x(0)=x_0\qquad x'(t)=\phi(t)X_3(x_\infty-x(t))\]
Since $\phi$ and $X_3$ are continuous, this equation has a unique solution by the Cauchy-Lipschitz theorem and one can check that the following
is a solution:
\[x(t)=x_\infty+\underbrace{\frac{x_0-x_\infty}{\sqrt{(e^{2\int_0^t\phi(u)du}-1)(1+(x_0-x_\infty)^2)+1})}}_{:=\alpha(x_0,x_\infty,t)}\]
Furthermore, one can check that for any $a,b\in\R$ and any $t>0$:
\begin{itemize}
\item $|\alpha(a,b,t)|\leqslant\frac{1}{\sqrt{e^{2\int_0^T\phi(u)du}-1}}$
\item $\min(0,a-b)\leqslant\alpha(a,b,t)\leqslant \max(0,a-b)$
\end{itemize}
It follows that:
\[\bar{g}-\eta-I(t)+\alpha(y(0),\bar{g}-\eta,t)\leqslant y(t)
\leqslant \bar{g}+\eta+I(t)+\alpha(y(0),\bar{g}+\eta,t)\]
\[-\eta-I(t)+\alpha(y(0),\bar{g}-\eta,t))\leqslant y(t)-\bar{g}\leqslant
\eta+I(t)+\alpha(y(0),\bar{g}+\eta,t)\]
Using the first inequality on $\alpha$ we get that:
\[-\eta-I(t)-\frac{1}{\sqrt{e^{2\int_0^T\phi(u)du}-1}}\leqslant y(t)-\bar{g}
\leqslant \eta+I(t)+\frac{1}{\sqrt{e^{2\int_0^T\phi(u)du}-1}}\]
Which proves the first result. And using the second inequality we get that:
\[-\eta-I(t)+\min(0,y(0)-(\bar{g}-\eta))|\leqslant y(t)-\bar{g}\leqslant
\eta+I(t)+\max(0,y(0)-(\bar{g}+\eta))\]
This proves the second result by case analysis.
\end{proof}

Sometimes though, the previous lemma lacks some precision. In particular when
$\phi$ is never close to $0$, where the intuition tells us that we should be
able to replace $\int_0^t|E(u)|du$ with some bound that does not depend on $t$.
The next lemma focuses on this case exclusively.

\begin{lemma}[Reach ODE: worst error]\label{lem:reach_ode2}
Let $T>0$, $I=[0,T]$, $g,E:I\rightarrow\R$, $\phi:I\rightarrow\Rps$.
Assume that there exist $\eta,\phi_{min},E_{max}>0$ and $\bar{g}\in\R$ such that
\begin{itemize}
\item For all $t\in I, |g(t)-\bar{g}|\leqslant\eta$.
\item For all $t\in I, |E(t)|\leqslant E_{max}$
\item For all $t\in I$, $\phi(t)\geqslant\phi_{min}$
\end{itemize}
Then the solution $y$ to \eqref{eq:reach} exists over $I$ and satisfies for all $t\in I$:
\[|y(t)-\bar{g}|\leqslant\eta+\frac{E_{max}}{\phi_{min}}+\frac{1}{\sqrt{\exp(2\int_0^t\phi(u)du)-1}}\]
\end{lemma}

\begin{proof}
Define $\psi(t)=\int_0^t\phi(u)du$ for $t\in I$. Since $\phi(t)\geqslant\phi_{min}>0$
then $\psi$ is an increasing function and admits an inverse $\psi^{-1}$. Define
for all $\xi\in[0,\psi(T)]$:
\[z_{\infty}(\xi)=g(\psi^{-1}(\xi))\quad\text{and}\quad z(\xi)=y(\psi^{-1}(\xi)).\]
One sees that $z$ satisfies
\[z'(\xi)=\underbrace{X_3(z_{\infty}(\xi)-z(\xi))+\frac{E(\psi^{-1}(\xi))}{\phi(\psi^{-1}(\xi))}}_{:=f(\xi,z(\xi))}\]
for $\xi\in[0,\psi(T)]$ and $z(0)=y(0)$. Furthermore, for all such $\xi$:
\[|z_{\infty}(\xi)-\bar{g}|\leqslant\eta\quad\text{ and }\quad
\left|\frac{E(\psi^{-1}(\xi))}{\phi(\psi^{-1}(\xi))}\right|\leqslant\frac{E_{max}}{\phi_{min}}.\]
Define $\alpha=\frac{E_{max}}{\phi_{min}}$,
\[f_{+}(x)=X_3(\bar{g}+\eta-x)+\alpha\quad\text{and}\quad
f_{-}(x)=X_3(\bar{g}-\eta-x)-\alpha.\]
One can check that $f_{-}(x)\leqslant f(\xi,x)\leqslant f_{+}(x)$ for any $\xi$ and $x$.
Consider the solutions $z_{-}$ and $z_{+}$ to $z_{-}'=f_{-}(z_{-})$ and $z_{+}'=f_{+}(z_{+})$
where $z_{-}(0)=z_{+}(0)=z(0)=y(0)$. By a classical result of differential inequalities,
we get that $z_-(\xi)\leqslant z(\xi)\leqslant z_+(\xi)$.
By shifting the solutions, both are instances of a system of the form:
\[x(0)=x_0\qquad x'(t)=-X_3(x(t))+\varepsilon\]
Since $x\mapsto-X_3(x)+\varepsilon$ is an increasing function, there exists a unique
$x_\infty$ such that $\varepsilon=X_3(x_\infty)$. Define $f(x)=-X_3(x)+\varepsilon$
and $f^{*}(x)=X_3(x_\infty-x)$. One checks that $f^{*}(x)-f(x)=3x_\infty(x^2-x_\infty^2)$,
thus $f^{*}(x)\leqslant f(x)$ if $x\leqslant x_\infty$ and $f(x)\leqslant f^{*}(x)$
if $x_\infty\leqslant x$. Notice that $f(x_\infty)=0$, so by a classical result
of differential equations, $x(t)-x_\infty$ must have a constant sign for the entire
life of the solution (i.e. $x(t)$ cannot ``cross'' $x_\infty$).
Consider the solutions $x_{-}$ and $x_{+}$ to $x_{-}=f^*(x_{-})$ and $x_{+}=f^*(x_{+})$
where $x_{-}(0)=\min(x_\infty, x_0)$ and $x_{+}(0)=\max(x_\infty,x_0)$.
Then the previous remark and a standard result guarantees that $x_{-}(t)\leqslant x(t)\leqslant x_{+}(t)$.
By the existence-uniqueness theorem for ODEs, the equations $x_{\pm}'=f^*(x_{\pm})$ have a unique solution and one can check that the following
are solutions:
\[x_{\pm}(t)=x_\infty+\frac{x_{\pm}(0)-x_\infty}{\sqrt{(e^{2t}-1)(1+(x_{\pm}(0)-x_\infty)^2)-1})}\]
We immediately deduce that
\[|x_{\pm}(t)-x_\infty|\leqslant\frac{1}{\sqrt{e^{2t}-1}}\]
and so
\[|x(t)-x_\infty|\leqslant\frac{1}{\sqrt{e^{2t}-1}}.\]
Let $\delta_\infty$ be such that $X_3(\delta_\infty)=\alpha$. Unrolling the definitions,
we get that
\[|z_{\pm}(\xi)-\bar{g}\mp\delta_\infty\mp\eta|\leqslant\frac{1}{\sqrt{e^{2t}-1}}.\]
So
\[|z(\xi)-\bar{g}|\leqslant\eta+\delta_\infty+\frac{1}{\sqrt{e^{2\xi}-1}}.\]
And finally, since $y(t)=z(\psi(t))$, we get that
\[|y(t)-\bar{g}|\leqslant\eta+\delta_\infty+\frac{1}{\sqrt{e^{2\int_0^t\phi(u)du}-1}}.\]
To conclude, it suffices to note that if $X_3(\delta_\infty)=\alpha$ then $\delta_\infty\leqslant\alpha$
since $X_3(x)\geqslant x$ for all $x$.
\end{proof}

\begin{definition}[Reach function]\label{def:reach}
For any $\phi\geqslant0$ and $y,g\in\R$, define
\[\reach(\phi,y,g)=2\phi X_3(g-y)\qquad\text{where }X_3(x)=x+x^3\]
\end{definition}

\begin{remark}\label{rem:reach_mult}
It is useful to note that for any $\phi,\psi\in\Rp$ and $y,g\in\R$,
\[\phi\reach(\psi,y,g)=\reach(\phi\psi,y,q)\]
\end{remark}

\begin{lemma}[Reach]\label{lem:reach}
There exists a function $\reach\in\gpval$ with the following property: 
given some arbitrary $I=[a,b]$, $\phi\in C^0(I,\Rp)$, $g,E\in C^0(I,\R)$, $y_0,g_\infty\in\R$ and $\eta>0$
such that for all $t\in I$, $|g(t)-g_\infty|\leqslant\eta$, let
$y:I\rightarrow\R$ be the solution of
\[
\left\{\begin{array}{@{}r@{}l}y(0)&=y_0\\y'(t)&=\reach(\phi(t),y(t),g(t))+E(t)\end{array}\right.
\]
Then for any $t\in I$,
\[|y(t)-g_\infty|\leqslant\eta+\int_a^t|E(u)|du+\exp\left(-\int_a^t\phi(u)du\right)
\quad\text{whenever }\int_a^t\phi(u)du\geqslant1\]
And for any $t\in I$,
\[|y(t)-g_\infty|\leqslant\max(\eta,|y(0)-g_\infty|)+\int_0^t|E(u)|du\]
\end{lemma}
\begin{proof}
Apply \lemref{lem:reach_ode} and notice that if $\int_a^t\phi(u)du\geqslant1$, then:
\begin{align*}
\sqrt{\exp\left(\int_a^t4\phi(u)du\right)-1}
&\geqslant\sqrt{(\exp\left(2\int_a^t\phi(u)du\right)+1)(\exp\left(2\int_a^t\phi(u)du\right)-1)}\\
&\geqslant\exp\left(\int_a^t\phi(u)du\right)\sqrt{e^2-1}\geqslant \exp\left(\int_a^t\phi(u)du\right)
\end{align*}
\end{proof}

\subsection{The proof}

We then get to the proof of $\gpwc \subseteq \gpoc$.

\begin{proof}
Apart from the issue of the input, the system is quite intuitive: we constantly feed
the \unaware{} system with the (smoothed) input and some precision. By increasing
the precision with time, we ensure that the system will converge when the input is stable.
However there is a small catch: over a time interval $I$, if we change the precision
within a range $[\check{\mu},\hat{\mu}]$ then we must provide the \unaware{} system
with precision based on $\hat{\mu}$ in order to get precision $\check{\mu}$.
Since the \unaware{} system takes time $\Omega(\infnorm{x},\hat{\mu})$
to compute, we need to make arrangements so that the requested precision doesn't change too much over periods
of this duration to make things simpler. We will use to our advantage
that $\Omega$ can always be assumed to be a constant.

Let $(f:\subseteq\R^n\rightarrow\R^m)\in\guc{\Upsilon}{\Omega}{\Lambda}{\Theta}$
where $\Upsilon,\Omega,\Lambda$ and $\Theta$ are polynomials, which we can assume
to be increasing functions of their arguments. Apply \lemref{lem:gpuc_time_rescaling} to
get $\omega>0$ such that for all $\alpha\in\R^n,\mu\in\Rp$:
\[\Omega(\alpha,\mu)=\omega\]
Apply \defref{def:guc} to get
$\delta,d$ and $g$.
Define:
\[\tau=\omega+2\qquad\delta'=\max(\delta,\tau+1)\]
Let $x\in C^0(\Rp,\R^n)$ and consider the following systems:
\[
\left\{\begin{array}{@{}r@{}l@{}}
x^*(0)&=0\\
y(0)&=0\\
z(0)&=0\\
\end{array}\right.
\qquad
\left\{\begin{array}{@{}r@{}l@{}}
{x^*}'(t)&=\reach(\phi(t),x^*(t),x(t))\\
y'(t)&=g(t,y(t),x^*(t),\mu(t))\\
z'(t)&=\sample_{[\omega+1,\omega+2],\tau}(t,\mu(t),z(t),y_{1..m}(t))
\end{array}\right.
\]
where
\[\phi(t)=\ln2+\mu(t)+\Lambda^*(2+x_1(t)^2+\cdots+x_n(t)^2,\mu(t))\qquad\mu(t)=\frac{t}{\tau}\]
Let $t\geqslant1$,
since $\phi\geqslant1$ then \lemref{lem:reach} gives:
\[
\infnorm{x^*(t)}\leqslant\pastsup{1}{\infnorm{x}}(t)+e^{-\int_{t-1}^t\phi(u)du}
    \leqslant\pastsup{1}{\infnorm{x}}(t)+1
\]
Also for $t\in[0,1]$ we get that:
\[\infnorm{x^*(t)}\leqslant\sup_{[0,t]}\infnorm{x}\]
This proves that $\infnorm{x^*(t)}\leqslant\pastsup{1}{\infnorm{x}}(t)+1$ for all $t\in\Rp$.
From this we deduce that:
\begin{align*}
\infnorm{y(t)}&\leqslant\Upsilon(\pastsup{\delta}\infnorm{x^*}(t),\pastsup{\delta}{\mu}(t),0)\\
    &\leqslant\poly(\pastsup{\delta}\infnorm{x}(t),t)
\end{align*}
Apply \lemref{lem:sample} to get that:
\begin{align*}
\infnorm{z(t)}&\leqslant2+\pastsup{\tau+1}\infnorm{y}(t)\\
    &\leqslant\poly(\pastsup{\delta'}\infnorm{x}(t),t)
\end{align*}
Let $I=[a,b]$ and assume there exist $\bar{x}\in\dom{f}$ and $\bar{\mu}$ such that
for all $t\in I$, $\infnorm{x(t)-\bar{x}}\leqslant e^{-\Lambda(\infnorm{\bar{x}},\bar{\mu})}$.
Note that
\[2+\sum_{i=1}^nx_i(t)^2\geqslant1+\infnorm{x(t)}\geqslant\infnorm{\bar{x}}\]
for all $t\in I$.
Let $n\in\N$ such that $n\geqslant\bar{\mu}+\ln2$ and $[n\tau,(n+1)\tau]\subseteq I$.
Note that $\mu(t)\in[n,n+1]$ for all $t\in I_n$.
Apply \lemref{lem:reach}, using that $\phi\geqslant1$,
to get that for all $t\in[n\tau+1,(n+1)\tau]$:
\begin{align*}
\infnorm{x^*(t)-\bar{x}}
    &\leqslant e^{-\Lambda^*(\infnorm{\bar{x}},n)}+e^{-\int_{n\tau}^t\phi(u)du}
    \leqslant 2e^{-\Lambda^*(\infnorm{\bar{x}},n)}\\
    &\leqslant e^{-\Lambda(\infnorm{\bar{x}},\bar{\mu}+\ln2)}
\end{align*}
Using the definition of \unaware{} computability, we get that:
\[\infnorm{y_{1..m}-f(\bar{x})}\leqslant e^{-\bar{\mu}+\ln2}\]
for all $t\in[n\tau+1+\omega,(n+1)\tau]=[n\tau+\omega+1,n\tau+\omega+2]$.

Define $J=[a+(1+\bar{\mu}+\ln2)\tau,b]\subseteq I$. Assume that $t\in J\cap[n\tau+1,(n+1)\tau]$
for some $n\in\N$, then we must have $(n+1)\tau\geqslant (1+\bar{\mu}+\ln2)\tau$
and thus $n\geqslant \bar{\mu}+\ln2$ so we can apply the above reasoning to get that
\[\infnorm{y_{1..m}(t)-f(x)}\leqslant e^{-\bar{\mu}+\ln2}.\]
Furthermore, we also have
\[\mu(t)\geqslant\frac{(1+\bar{\mu}+\ln2)\tau}{\tau}\geqslant \bar{\mu}+\ln2\]
for all $t\in J$.
Apply \lemref{lem:sample} to conclude that for any $t\in [a+\tau+\bar{\mu}+\ln2+\tau+1,b]$,
we have
\[\infnorm{z(t)-f(x)}\leqslant 2e^{-\bar{\mu}+\ln2}\leqslant e^{-\bar{\mu}}.\]

To conclude the proof, we need to rewrite the system as a PIVP using
\lemref{prop:gpac_ext_ivp_stable_pre}.
Note that this works because we only rewrite the variable $y$, and doing so we
require that $x^*$ be a $C^1$ function (which is the case) and the new initial
variable will depend on $x^*(0)=0$ which is constant.
\end{proof}

\section{Proof that AOP implies ATSP}

The purpose of the current section is to show one last inclusion which, in conjunction with all
the inclusions of the previous sections, closes the circle of inclusions and shows Theorem
\ref{th:main}.

\begin{theorem}
$\gpoc\subseteq\gpc$.
\end{theorem}

\begin{proof}
The proof is trivial: given $x$, we store it in a variable and run the online system.
Since the input has no error, we can directly apply the definition to get that
the online system converges.

Let $(f:\subseteq\R^n\rightarrow\R^m)\in\goc{\Upsilon}{\Omega}{\Lambda}$. Apply
\defref{def:goc} to get $\delta,d,p$ and $y_0$. Let $x\in\dom{f}$ and consider the
following system:
\[
\left\{\begin{array}{@{}r@{}l@{}}
x(0)&=x\\
y(0)&=y_0\\
\end{array}\right.
\qquad
\left\{\begin{array}{@{}r@{}l@{}}
x'(t)&=0\\
y'(t)&=p(y(t),x(t))\\
\end{array}\right.
\]
We immediately get that:
\[\infnorm{y(t)}\leqslant\Upsilon(\pastsup{\delta}{\infnorm{x}}(t),t)\leqslant\Upsilon(\infnorm{x},t)\]
Let $\mu\in\Rp$ and let $t\geqslant\Omega(\infnorm{x},\mu)$, then apply \defref{def:goc}
to $I=[0,t]$ to get that
\[\infnorm{y_{1..m}(t)-f(x)}\leqslant e^{-\mu}\]
since $\infnorm{x(t)-x}=0$.
\end{proof}

\section{Conclusion}

As a conclusion, we proved actually even a stronger statement than
\thref{th:main}, namely:

\begin{theorem} \label{th:main1}
All notions of computations are equivalent, both at the computability
level:
$$\cglc=\cgc=\cgwc=\cgoc$$
and at the complexity level:
$$\gplc=\gpc=\gpwc=\gpoc$$
\end{theorem}

\section*{References}

\bibliographystyle{elsarticle-harv}
\bibliography{
  ContComp
}

\end{document}